\documentclass[11pt]{article}
\usepackage{helvet}

\usepackage{fancyhdr}

\usepackage{amssymb,amsbsy,amsfonts,amsmath,amsthm,xspace}
\usepackage[top=1in, bottom=1in, left=1in, right=1in]{geometry}

\usepackage{amssymb,amsbsy,amsfonts,amsmath,amsthm,xspace}
\usepackage{mathrsfs}
\usepackage{courier}
\usepackage[utf8]{inputenc}
\usepackage{tikz}
\usetikzlibrary{calc, matrix, arrows, shapes, positioning, fit, shapes.misc, shapes.geometric, decorations.pathreplacing, patterns}
\usepackage{booktabs}
\usepackage{sectsty}
\sectionfont{\fontsize{11}{11}\selectfont}
\subsectionfont{\fontsize{11}{11}\selectfont}
\subsubsectionfont{\fontsize{11}{11}\selectfont}
\usepackage{graphicx}
\usepackage{caption}
\usepackage{subcaption}
\usepackage{setspace}
\usepackage{enumitem}
\usepackage{colortbl}
\usepackage{multirow}
\usepackage{epstopdf}
\usepackage{epsfig}
\usepackage{hyperref}
\usepackage{url}
\usepackage[toc,page]{appendix}
\usepackage{float}
\usepackage{natbib}
\usepackage{color}
\usepackage{verbatim}
\usepackage{authblk}
\usepackage{wrapfig}
\usepackage{algorithm}
\usepackage{algpseudocode}
\usepackage[normalem]{ulem}
\usepackage{xparse}
\usepackage{tabularx}
\usepackage{rotating}
\usepackage{makecell}
\usepackage{titlesec}
\usepackage{cleveref}

\titlespacing*{\section}{0pt}{10pt}{5pt}
\titlespacing*{\subsection}{0pt}{6pt}{3pt}

\def\text#1{\mbox{\rm #1}}

\def\dfrac{\displaystyle\frac}

\hypersetup{
  colorlinks   = true,
  urlcolor     = blue,
  linkcolor    = blue,
  citecolor    = blue
}

\newcommand{\N}{\mathcal{N}}

\tikzset{baseline adjust/.style={baseline={([yshift=-0.5ex]current bounding box.center)}}}
\algrenewcommand\algorithmicrequire{\textbf{Input:}}
\algrenewcommand\algorithmicensure{\textbf{Output:}}

\theoremstyle{definition}
\newtheorem{example}{Example}

\newtheorem{theorem}{Theorem}
\newtheorem{lemma}{Lemma}

\newtheorem{corollary}{Corollary}
\newtheorem{rem}{Remark}
\newtheorem{definition}{Definition}
\newtheorem{assumption}{Assumption}

\crefname{theorem}{Theorem}{Theorems}
\Crefname{theorem}{Theorem}{Theorems}
\crefname{lemma}{Lemma}{Lemmas}
\Crefname{lemma}{Lemma}{Lemmas}
\crefname{proposition}{Proposition}{Propositions}
\Crefname{proposition}{Proposition}{Propositions}
\crefname{corollary}{Corollary}{Corollaries}
\Crefname{corollary}{Corollary}{Corollaries}
\crefname{definition}{Definition}{Definitions}
\Crefname{definition}{Definition}{Definitions}
\crefname{assumption}{Assumption}{Assumptions}
\Crefname{assumption}{Assumption}{Assumptions}

\usepackage{xcolor}
\usepackage[normalem]{ulem}
\newif\ifreview
\reviewfalse               

\DeclareRobustCommand{\revadd}[1]{%
  {\ifreview\color{red}\fi#1}%
}
\DeclareRobustCommand{\revdel}[1]{%
  \ifreview
    {\sout{#1}}%
  \fi
}

\begin{document}

\begin{center}
\textbf{\large Nonparametric estimation of time-varying network connections by multi-stage smoothing}\\
\bigskip
Jeonghwan Lee\\
Tianxi Li and Adam J. Rothman\\
\medskip
University of Minnesota, Twin Cities
\end{center}

\begin{abstract}
We consider the problem of estimating the underlying edge probabilities of a time-varying network observed at multiple time points. The probability structure is represented by a time-varying graphon that satisfies temporal H\"{o}lder smoothness and piecewise Lipschitz conditions in the latent variables. We propose a multi-stage smoothing estimator that first applies temporal local smoothing to each edge and then performs node-domain smoothing using a data-driven neighborhood construction. An additional temporal smoothing step is introduced as an optional refinement when uniform accuracy over the entire time domain is required. Simulation studies demonstrate the benefits of combining temporal and node-domain smoothing under different generative models. We also apply the method to a real time-varying network dataset and show that it captures both smooth temporal evolution and structural patterns in the connectivity.
\end{abstract}

\section{Introduction}\label{sec:intro}
Statistical network modeling provides a principled framework for analyzing complex relational systems, particularly in settings where interactions evolve over time. Time-varying networks arise in a variety of ubiquitous applications, such as functional brain connectivity \citep{thompson2017static, zhang2020mixed}, gene and genomic regulatory processes \citep{zhang2017finding, bartlett2021two}, and social or economic environments \citep{snijders2010maximum, kolar2010estimating}. In these contexts, measurements collected at different time points record how observed connections fluctuate, forming a sequence of network snapshots that reflect the temporal evolution of the underlying system. For example, fMRI studies yield time-indexed measurements of activity across brain regions, from which researchers construct connectivity networks that change over the scanning period \citep{bassett2011dynamic, rubinov2010complex}. Similarly, in political systems such as the U.S.\ Senate, legislative cosponsorship records give rise to network snapshots that naturally vary across sessions \citep{fowler2006connecting, kirkland2014measurement}. General reviews of time-varying network analysis, including methodological developments and representative applications, are provided in \cite{holme2012temporal} and \cite{kim2018review}.

In this paper, we consider the setting where the connections of a fixed set of $n$ individuals are observed over time. The connections at a given time point are recorded as an unweighted and undirected network, represented by an $n\times n$ adjacency matrix $A^{(t)}$, where $A^{(t)}_{ij}=1$ if nodes $i$ and $j$ are linked at time $t$, and $0$ otherwise. Given observation time points $t_1,\ldots,t_m$, the sequence of time-varying networks can therefore be represented as $\{A^{(t_k)}\}_{k=1}^m$. A commonly used statistical framework for modeling a network at a specific time is the inhomogeneous Erdős–Rényi model. Under this framework, there exists a probability matrix $P \in [0,1]^{n \times n}$, and conditionally on $P$, each edge is established independently as $\mathrm{Bernoulli}(P_{ij})$. To accommodate temporal evolution, we adopt a natural extension of this framework by assuming $A_{ij}^{(t)} \sim \mathrm{Bernoulli}\bigl(P_{ij}^{(t)}\bigr)$, allowing the connection probabilities to change over time. Modeling time-varying networks under this framework thus reduces to modeling and estimating the sequence of probability matrices $P^{(t)}$.

Extensive efforts in the literature have explored time-varying network estimation under various parametric models. A prominent line of work builds on dynamic generalizations of the classical stochastic block model (SBM) \citep{holland1983stochastic}. For instance, \cite{yang2011detecting} and \citet{matias2017statistical} modeled the temporal evolution of block memberships via Markov chains, with the latter further allowing the connectivity parameters to vary over time and establishing consistency guarantees for likelihood-based estimators. Other formulations allow the community-connection probabilities to vary, such as the state-space approach of \citet{xu2014dynamic} or the piecewise-constant evolution model of \citet{bhattacharjee2020change}. Additional structural conditions, such as sparsity or temporal smoothness, have been incorporated to facilitate estimation in high-dimensional settings \citep{keriven2022sparse}. Related autoregressive network formulations model temporal dependence at the edge level rather than through block structures \citep{jiang2023autoregressive}. While dynamic SBMs provide interpretable low-dimensional summaries, the imposed block partitions restrict these models to group-level patterns of temporal evolution.

The latent space models \citep{hoff2002latent}, which treat each node as a point in a low-dimensional Euclidean space, are another widely used model class in studying real-world networks \citep{ma2020universal,zhang2022joint,liu2025grand}. As one of the earliest dynamic latent space models, \citet{sarkar2006dynamic} modeled latent positions as evolving according to Gaussian random-walk processes across successive network snapshots. Subsequent work formalized dynamic latent space models within explicit Markov state-space frameworks \citep{sewell2015latent}, while other approaches considered smoothly varying latent trajectories \citep{macdonald2025latent} within random dot product graphs \citep{rubin2022statistical}. In a Bayesian direction, \citet{loyal2025generalized} proposed a generalized Bayesian framework for dynamic random dot product graphs, providing posterior-based uncertainty quantification for time-evolving latent embeddings. More recently, \citet{athreya2025euclidean} developed a Euclidean-mirror metric to provide a geometric framework for analyzing dynamics in latent-position models.

While dynamic SBMs and latent space models introduce useful structural assumptions, they heavily rely on parametric assumptions—either a finite block partition or a low-dimensional embedding. Such constraints can restrict their flexibility and limit their ability to capture fine-grained, complex connection patterns. In contrast, the graphon model \citep{bickel2009nonparametric} provides a fully nonparametric framework for modeling the connection probability matrix. A time-varying version of the graphon model was first formulated by \citet{pensky2019dynamic}, where $P^{(t)}$ is modeled as an instantiation of a time-varying graphon function that is smooth over time (formally defined in Definition~\ref{def:graphon}). However, while they studied the minimax theory, they did not provide a computationally feasible estimation method. Alternatively, \citet{zhao2019change} developed a graphon-based change-point detection method for independent networks across time, but their model restricts the graphon to remain strictly static between change points, precluding temporal evolution.

Motivated by the generality and flexibility of the graphon, we adopt the time-varying graphon model as our framework. To address the aforementioned computational and structural limitations, we propose a computationally efficient, multi-stage smoothing estimator that integrates time-domain smoothing with node-domain smoothing to recover the evolving edge probabilities. Furthermore, we establish finite-sample estimation error bounds for the proposed estimator, including a uniform convergence rate over time. Through extensive simulation studies, we evaluate its performance across a range of time-varying network generative models and demonstrate its substantial advantages over several benchmark methods. We also apply the method to U.S.\ Senate cosponsorship networks from 1973 to 2024, illustrating that our approach successfully recovers highly interpretable, node-level temporal dynamics.

The remainder of this paper is organized as follows. Section~\ref{sec:method} introduces the time-varying graphon framework and details the proposed multi-stage smoothing estimators. Section~\ref{sec:theory} establishes the main theoretical properties, including finite-sample error bounds. Section~\ref{sec:simulation} presents simulation studies under various time-varying network generative models. In Section~\ref{sec:data}, we apply our methodology to the U.S.\ Senate cosponsorship data. Finally, Section~\ref{sec:discussion} concludes the paper with a discussion of limitations and directions for future research.

\section{Methodology} \label{sec:method}
\subsection{Sparse time-varying graphon model}
Our primary objective is to estimate the sequence of probability matrices $\{P^{(t)}\}_{t\in [0,1]}$ that govern a time-varying network. We formally define the time-varying graphon model below.

\begin{definition}[Sparse Time-Varying Graphon Model \citep{bickel2009nonparametric}]\label{def:graphon}
Suppose there exists a measurable function $f : [0,1]^3 \to [0,1]$ satisfying $f(x,y,t)=f(y,x,t)$ for all $(x,y,t) \in [0,1]^3$. We further assume $\max f(x,y,t)= 1$. A size $n$ time-varying network $\{A^{(t)}\}_{t\in[0,1]}$ is said to follow a sparse time-varying graphon model if there exists a scalar $\rho_n \in (0,1)$, and a set of independent and identically distributed $\text{Uniform}[0,1]$ random variables $\xi_1,\ldots, \xi_n$, such that for each time $t$, the probability matrix $P^{(t)} \in [0,1]^{n\times n}$ is specified by
\ifreview
\begin{align*}
    \text{\revdel{\mbox{$
    P_{ij}^{(t)}=\rho_n f(\xi_i,\xi_j,t),
    \qquad 1\le i,j\le n,\quad t\in[0,1].
    $}}}
\end{align*}
\fi
\begin{align*}
    \revadd{
    P_{ij}^{(t)}=
    \begin{cases}
        \rho_n f(\xi_i,\xi_j,t), & i\neq j,\\
        0, & i=j,
    \end{cases}
    \qquad 1\le i,j\le n,\quad t\in[0,1].
    }
\end{align*}
\end{definition}
Given $P^{(t)}$, the observed adjacency matrices are generated by independent Bernoulli trials,
\begin{align*}
    A_{ij}^{(t)} \sim {\rm Bernoulli}\bigl(P_{ij}^{(t)}\bigr), 
    \qquad A_{ij}^{(t)} = A_{ji}^{(t)}, \text{~~for all~~} i<j.
\end{align*}
\noindent
Note that the scaling factor $\rho_n$ is introduced to accommodate sparse networks. From a theoretical perspective, we are particularly interested in the regime where $\rho_n \to 0$ as $n\to \infty$.

Within this framework, we focus on graphons that vary smoothly
over time. A wide range of existing network models admit such
representations, including community models and latent space
models with smoothly evolving parameters. The following examples
illustrate the flexibility of the modeling framework; additional
regularity conditions for our estimation guarantees are introduced in Section~\ref{sec:theory}.

\begin{example} \label{ex:sbm}
\textbf{Dynamic networks with community structures.} Assume that the $n$ nodes are partitioned into $K$ communities. Let $G \in \{0,1\}^{n \times K}$ denote the community membership matrix, assumed fixed over time, and let $B^{(t)} \in [0,1]^{K \times K}$ denote the time-varying community connectivity matrix. \revadd{This fixed-membership setting is a restricted version of the dynamic SBM studied by \citet{pensky2019spectral}, whose model also allows nodes to change their community memberships over time.} The dynamic stochastic block model specifies the edge probability matrix as \revdel{$P^{(t)} = G B^{(t)} G^\top$} 
\revadd{$P^{(t)} = \rho_n G B^{(t)} G^\top$}, while its degree-corrected extension introduces a diagonal matrix $\Theta^{(t)}$ of node-specific degree parameters and takes the form
\revdel{$P^{(t)} = \Theta^{(t)} G B^{(t)} G^\top \Theta^{(t)}$}
\revadd{$P^{(t)} = \rho_n \Theta^{(t)} G B^{(t)} G^\top \Theta^{(t)}$}. These models correspond to time-varying graphons of the form $f(\xi_i,\xi_j,t)=B^{(t)}_{kl}$ or $f(\xi_i,\xi_j,t)=\theta_i^{(t)}\theta_j^{(t)}B^{(t)}_{kl}$ when nodes $i$ and $j$ belong to communities $k$ and $l$, respectively.

Allowing for temporal variation of community membership is feasible under the time-varying graphon model, via the more flexible temporal degree-corrected mixed membership block model (DCMM) introduced in \citep{ji2022cocitation}, the geneneralization of the static DCMM of \cite{jin2024mixed}. The model specifies $P_{ij}^{(t)}=\rho_n\theta_i^{(t)}\theta_j^{(t)}(\pi_i^{(t)})^\top B\pi_j^{(t)}$ for $i\neq j$, where $\theta_i^{(t)}$ and $\pi_i^{(t)}$ are time-varying degree parameters and community-membership probability vectors, respectively, and $B$ is a time-invariant block connectivity matrix. When the degree parameters and community-membership vectors vary smoothly over time, the temporal DCMM provides a block-structured example of a smoothly time-varying graphon model.
\end{example}

\begin{example}\label{ex:lsm}
\textbf{Dynamic latent space model.} A dynamic latent position model assumes time-indexed latent trajectories $Z_i^{(t)}\in\mathbb{R}^d$, where $d$ is the latent dimension of the embedding, and a link function $\kappa$ so that \revdel{$P_{ij}^{(t)} = \kappa(Z_i^{(t)}, Z_j^{(t)})$, yielding a time-varying graphon $f(\xi_i,\xi_j,t)=\kappa(X_i^{(t)},X_j^{(t)})$.} \revadd{$P_{ij}^{(t)} = \rho_n\kappa(Z_i^{(t)},Z_j^{(t)})$ for $i\neq j$. To connect these trajectories to the graphon formulation, let $g:[0,1]\times[0,1]\to\mathbb{R}^d$ be a measurable function and define $Z_i^{(t)}=g(\xi_i,t)$, where $\xi_1,\ldots,\xi_n$ are independent $\mathrm{Uniform}[0,1]$ variables. This yields the time-varying graphon $f(x,y,t)=\kappa(g(x,t),g(y,t))$. The latent positions can evolve over time through $g(\xi_i,t)$ even though each $\xi_i$ remains fixed.} \revdel{They provide a smooth example: latent processes $Z_i^{(t)}$ with inner-product structure $P^{(t)}=Z^{(t)}Z^{(t)\top}$, which implies $P_{ij}^{(t)}=Z_i^{(t)\top}Z_j^{(t)}$. Since $Z_i^{(t)}$ evolves smoothly in $t$, their model fits naturally within our smooth time-varying graphon formulation.} \revadd{An important instance is the Bernoulli inner-product structure considered by \citet{macdonald2025latent}. Taking $\kappa(z,z')=z^\top z'$ on an admissible latent space gives the graphon $f(x,y,t)=g(x,t)^\top g(y,t)$. The same framework also accommodates the latent distance model used in our simulations, where $P_{ij}^{(t)}=\operatorname{logit}^{-1}\bigl(\alpha-\|g(\xi_i,t)-g(\xi_j,t)\|_2\bigr)$ for $i\neq j$.}
\end{example}

\subsection{Nonparametric estimation by multi-stage smoothing}

Given the generality of the time-varying graphon framework, the model estimation can be challenging and additional regularity conditions will be needed. Our estimator combines smoothing over time with smoothing across
nodes, relying on 1) the temporal smoothness: the graphon $f$ varies smoothly across the time dimension $t$, and 2) the ``node-wise roughness'' of the graphon: $f$ is piecewise smooth in its first two arguments, allowing for occasional abrupt jumps. These assumptions will be formally stated in Assumptions~\ref{assump:smooth} and \ref{assump:lipschitz} of Section~\ref{sec:theory}. 

In practice, the network is recorded only at a finite collection of time points $t_1,\ldots,t_m$, yielding the observed snapshot sequence $\{A^{(t_k)}\}_{k=1}^m$. Our objective is to estimate the probability matrix $P^{(t)}$ for an arbitrary $t \in [0,1]$---not merely at the sampled time points. Because we assume the model exhibits smoothness both temporally and across nodes, it is natural to incorporate smoothing in both domains. Specifically, we employ a multi-stage smoothing procedure that combines local polynomial smoothing over the temporal domain with neighborhood smoothing over the node domain. However, since the latent variables $\xi_i$ are unobserved, temporal and node-domain smoothness cannot be leveraged symmetrically. Consequently, for any given time point $t\in [0,1]$, we adopt a principled, two-stage sequential smoothing procedure, detailed below.

\textbf{Step 1: Temporal local polynomial smoothing.}
For each node pair $(i,j)$, we construct an intermediate estimate of the probability trajectory $t \mapsto P_{ij}^{(t)}$ by applying a local polynomial smoother to the sequence $\{A_{ij}^{(t_k)}\}_{k=1}^m$. 
The estimator takes the form
\begin{align} \label{eq:local-poly}
\Tilde P_{ij}^{(t)} = \sum_{k=1}^m w_k(t; h_1)\,A_{ij}^{(t_k)}
\end{align}
where the weights $w_k(t; h_1)$ correspond to the equivalent kernel weights derived from an $\ell$th-degree local polynomial regression fit at $t$, parameterized by a bandwidth $h_1$       and the polynomial order $\ell$ \citep{Tsybakov}. By design, network snapshots closer to $t$ receive higher weights. This step explicitly exploits the temporal smoothness of $f(\xi_i,\xi_j,t)$ to mitigate the noise inherent in individual snapshots.

\textbf{Step 2: Node-domain neighborhood smoothing.}
The intermediate estimator $\tilde{P}^{(t)}$ serves as a denoised starting point for approximating the true $P^{(t)}$. We subsequently incorporate node-domain smoothing by borrowing strength from nodes with similar connection profiles, using the neighborhood smoothing procedure of \citet{zhang2017estimating}. Specifically, for each node $i$, we define its empirical distance to another node $i'$ as
\begin{align*}
    \Tilde d_t(i,i') 
    = 
    \sqrt{\max_{l \neq i,i'} 
    \bigl|
       \langle\,\Tilde P_{i.}^{(t)} - \Tilde P_{i'.}^{(t)},\, \Tilde P_{l.}^{(t)}\rangle
    \bigr| \;/\; n},
\end{align*}
which compares the difference between two nodes' temporal probability vectors via their inner products with all other rows. Intuitively, this distance quantifies the similarity between their respective graphon slices $f(\xi_i,\cdot,t)$ and $f(\xi_{i'},\cdot,t)$. It is worth noting that two nodes with a small empirical distance do not necessarily have latent positions $\xi_i$ and $\xi_{i'}$ close to each other. However, provided their graphon slices are similar, it is statistically justified to leverage the connection patterns of $i'$ to help the probability estimates for $i$. A more detailed theoretical discussion of this distance metric is provided by \cite{zhang2017estimating}.

Building on this intuition, the neighborhood for node $i$ at time $t$ is defined as
\begin{align} \label{eq:nbhd}
\mathcal N_i^{(t)} = \{i' : \Tilde d_t(i,i') \le q_{i,t}(h_2)\},
\end{align}
where $q_{i,t}(h_2)$ is the $h_2$-quantile of the distance set $\{\Tilde d_t(i,k)\}_k$. 
\revdel{The refined estimator of $P_{ij}^{(t)}$ is then constructed by averaging the intermediate estimates within this neighborhood:}
\revadd{We define the row-wise two-stage estimator by averaging the temporally smoothed estimates over the neighborhood of each node:}
\begin{align}\label{eq:Phat}
\hat P_{ij}^{(t)} =
\frac{1}{|\N_i^{(t)}|}
\sum_{i' \in \N_i^{(t)}} \Tilde P_{i'j}^{(t)},
\quad
\ifreview \text{\revdel{\mbox{$i<j$}}}\;\fi
\revadd{i\neq j,\qquad \hat P_{ii}^{(t)}=0}.
\end{align}
\revadd{Although the target matrix $P^{(t)}$ is symmetric, $\hat P^{(t)}$ need not be symmetric because each row is estimated using its own neighborhood. To obtain a symmetric estimator, we define
\begin{equation} \label{eq:sym}
\hat P_{Sym}^{(t)}=\frac{1}{2}
\left\{
\hat P^{(t)}+\bigl(\hat P^{(t)}\bigr)^\top
\right\}.
\end{equation}
Our theoretical analysis establishes a maximum-row error bound for the row-wise estimator $\hat P^{(t)}$ and derives a Frobenius error bound for its symmetric counterpart $\hat P_{Sym}^{(t)}$.}

\begin{rem}\label{rem:smoothing-order}
Our proposed estimation procedure purposefully executes temporal smoothing prior to node-domain neighborhood smoothing. A natural question arises: \emph{why not reverse this order?} This design choice is rooted in the following intuition. In practice, raw network connections are highly noisy, particularly in sparse regimes. Because the latent variables $\xi_i$ are unobserved, neighborhood smoothing relies entirely on the empirical data to first estimate the neighborhoods before applying the smoothing step. If applied directly to the raw, noisy snapshots, this neighborhood identification step incurs severe instability over time. Consequently, a follow-up temporal smoothing step would aggregate values over substantially fluctuating neighborhoods, introducing large and unnecessary biases into the final estimator. In contrast, our current design leverages temporal smoothing to properly denoise individual entries over time without inducing such structural biases. The subsequent neighborhood smoothing step then operates on this denoised matrix, yielding a far more stable and accurate neighborhood identification. The empirical advantage of this sequential design is explicitly verified in our simulation experiments in Section~\ref{sec:simulation}. 
\end{rem}

\textbf{Step 3: Temporal refinement smoothing.}
For most applications, the estimator $\hat P^{(t)}$ obtained from the first two stages is sufficiently accurate. However, because the estimated neighborhoods $\N_i^{(t)}$ may shift discretely with $t$, the resulting sequence $\{\hat P^{(t_k)}\}_{k=1}^m$ is not strictly guaranteed to be smooth over time. \revdel{When strict temporal smoothness is required, we can enforce it by applying a final refinement operation to the sequence $\{\hat P^{(t)}\}$.} \revadd{When strict smoothness is required, we can apply a final refinement operation to the sequence $\{\hat P^{(t_k)}\}_{k=1}^m$.} We define the refined estimator as
\begin{equation}\label{eq:third-stage}
\bar P^{(t)}
=
\sum_{k=1}^m w_k(t;h_3)\,\hat P^{(t_k)},
\end{equation}
\revdel{using the same local polynomial weighting scheme as in Step~1, but parameterized by a different bandwidth $h_3$. This additional smoothing step guarantees a smooth estimator over time. From a theoretical standpoint, this enforced smoothness is crucial for establishing uniform error bounds.} \revadd{using the same local polynomial weighting scheme, kernel, and polynomial order as in Step~1. We take $h_3=h_1$, reusing the first-stage bandwidth. With a sufficiently smooth kernel, this additional smoothing step yields a smooth estimator over time. From a theoretical standpoint, this refinement allows us to establish uniform error bounds over $t\in[0,1]$, as shown in Theorem~\ref{thm:three-step}.} Empirically, however, we observe that \revdel{$\bar P_{Sym}^{(t)}$ }the additional refinement step yields only marginal improvements in estimation error, if any. Consequently, we use the two-stage estimator \revadd{$\hat P_{Sym}^{(t)}$} for all empirical evaluations in this paper.

\subsection{Parameter tuning by cross-validation}
In practice, we have to specify the temporal bandwidth $h_1$, the polynomial degree $\ell$, and the neighborhood parameter $h_2$. We propose to employ a cross-validation strategy. 

While using cross-validation to tune neighborhood smoothing on static networks can be done by the edge cross-validation (ECV) method of \cite{li2020network}, directly adapting ECV to our time-varying setting will introduce an intensive computational burden that is practically prohibitive. To circumvent this computational bottleneck, we instead adopt a highly scalable leave-one-time-out cross-validation strategy for the current data structure. Specifically, given the observations $\{A^{(t_k)}\}_{k=1}^m$, one network snapshot at $t_k$ is  held out as the validation set, while the remaining matrices $\{A^{(t_s)} : s \neq k\}$ serve as the training data. For each candidate parameter combination $\theta = (\ell, h_1, h_2)$ drawn from a predefined grid (with $h_1$ and $h_2$ typically spaced on a logarithmic scale), we apply our estimation procedure to the $m-1$ training snapshots and evaluate the model's performance in predicting the unobserved network connections at the validation time point. This procedure is described in Algorithm~\ref{alg:cv}.





\begin{algorithm}[ht]
\caption{Cross-validation procedure for selecting $(\ell, h_1, h_2)$}
\label{alg:cv}
\begin{algorithmic}[1]

\Require Time-varying network $\{A^{(t_k)}\}_{k=1}^m$, 
candidate tuning parameters $\theta$
\ForAll{$\theta = (\ell, h_1, h_2) \in \Theta$}
    \For{$k = 1$ to $m$}
    
        \State Construct training set 
        $A^{(-t_k)} = \{A^{(t_s)} : s \neq k\}$
        
        \State $\tilde{P}^{(t_k)} \gets 
        \textsc{LocalPolynomial}(A^{(-t_k)};\ell,h_1,t_k)$
        
        \State $\hat{P}_{Sym}^{(t_k)} 
        \gets \textsc{nSmooth}(\tilde{P}^{(t_k)};h_2)$
        
        \State Compute validation error
        \[
        \mathbf{Err}_{\mathrm{CV}}(\theta; t_k)
        = 
        \frac{\|\hat{P}_{Sym}^{(t_k)} - A^{(t_k)}\|_F}
             {\|A^{(t_k)}\|_F}
        \]
    \EndFor
    
    \State 
    \[
    \mathbf{Err}_{\mathrm{Mean}}(\theta)
    =
    \frac{1}{m}
    \sum_{k=1}^{m}
    \mathbf{Err}_{\mathrm{CV}}(\theta; t_k)
    \]
    
\EndFor

\State 
\[
\hat \theta =
\arg\min_{\theta}
\mathbf{Err}_{\mathrm{Mean}}(\theta)
\]
\Statex
\Statex \textbf{where}
\Statex $\textsc{LocalPolynomial}(A;\ell,h_1,t)$ denotes 
the local polynomial estimator defined in \eqref{eq:local-poly}, 
computed using adjacency matrices $A$ and evaluated at time $t$,
and $\textsc{nSmooth}(\cdot;h_2)$ implements 
Step~2 (node-domain smoothing) as defined in \eqref{eq:sym}.
\end{algorithmic}
\end{algorithm}

\section{Theoretical properties of the multi-stage smoothing estimator} \label{sec:theory}
In this section we study the theoretical properties of the proposed estimators. For a matrix $B \in \mathbb{R}^{n \times n}$, we use $\|B\|_F$ to denote its Frobenius norm $\|B\|_F = \left( \sum_{i,j} B_{ij}^2 \right)^{1/2}$, and $\|B\|_{2,\infty}$ to denote its $L_{2,\infty}$ operator norm $\|B\|_{2,\infty} = \max_{1 \le i \le n} \|B_{i\cdot}\|_2$, i.e., the maximum Euclidean norm of the rows of $B$. \revadd{We write $a_{n,m}\lesssim b_{n,m}$ if $a_{n,m}\le Cb_{n,m}$ for some constant $C>0$ independent of $n$ and $m$, for all sufficiently large $n,m$. Similarly, $a_{n,m}\gtrsim b_{n,m}$ means $b_{n,m}\lesssim a_{n,m}$, and $a_{n,m}\asymp b_{n,m}$ means that both inequalities hold.}

We begin by formally defining the regularity conditions for the time-varying graphon model.

\begin{assumption}{(H\"older smooth graphon)} \label{assump:smooth}
Let $f: [0,1]^3 \to [0,1], \; (\xi_i,\xi_j,t) \mapsto f(\xi_i,\xi_j,t)$, be a time-varying graphon, where $\xi_i,\xi_j \in [0,1]$ denote the latent positions of nodes $i$ and $j$, respectively. For each fixed pair $(\xi_i,\xi_j)\in [0,1]^2$, assume the function $t \mapsto f(\xi_i,\xi_j,t)$ belongs to the H\"older class $\mathbb{H}^\beta(L_1)$ on $[0,1]$, with smoothness parameter $\beta > 0$ and H\"older constant $L_1 > 0$. Specifically, for $\ell = \lfloor \beta \rfloor$, we assume $t \mapsto f(\xi_i,\xi_j,t)$ is $\ell$-times differentiable and its $\ell$-th derivative satisfies
\begin{equation*}
\bigl|f^{(\ell)}(\xi_i,\xi_j,t) - f^{(\ell)}(\xi_i,\xi_j,t')\bigr| \;\le\; L_1\,|t - t'|^{\beta - \ell},\quad\forall\,t,t'\in[0,1].
\end{equation*}
\end{assumption}

\noindent
Under this assumption, the network snapshots at different time points are observed independently, while the underlying edge probabilities $P_{ij}^{(t)}$ evolve smoothly in $t$. 

\begin{assumption}{(Piecewise Lipschitz graphon)} \label{assump:lipschitz}
There exist a constant $\delta>0$ and an integer $J \ge 1$ and a partition $0 = x_0 < x_1 < \cdots < x_J = 1$
    such that 
    \[
    \min_{0 \,\leq\, s \,\leq\, J-1} \bigl(x_{s+1} - x_s\bigr) \;>\; \delta
    \] 
 such that  the graphon $f$ is Lipschitz in its first two arguments with constant $L_2$ within each sub-interval $[x_k, x_{k+1})$. Specifically, for each $0 \le k,l \le J-1$ and any $t\in [0,1]$, for all 
    $u, u_1,u_2 \in [x_k,x_{k+1})$
    and 
    $
    v, v_1,v_2 \in [x_l,x_{l+1})
    $,
    \begin{equation*}
        \bigl|f(u_1,v,t) - f(u_2,v,t)\bigr| \;\le\; L_2\,|u_1 - u_2|
    \quad\text{and}\quad
    \bigl|f(u,v_1,t) - f(u,v_2,t)\bigr| \;\le\; L_2\,|v_1 - v_2|.
    \end{equation*}
\end{assumption}
\revdel{As illustrated in Examples~{\ref{ex:sbm}} and {\ref{ex:lsm}}, several dynamic network models satisfy the previous two assumptions.} We next formalize the regularity conditions of the local polynomial estimator. For each fixed $t\in[0,1]$, the temporal local polynomial estimator admits the representation given in \eqref{eq:local-poly}, where $w_k(t;h_1)$ are the equivalent kernel weights of a local polynomial of order $\ell$. Following the notations in \cite{Tsybakov}, these weights are determined by the local polynomial design matrix 
$B_{mt}
=
\frac{1}{m h_1}
\sum_{k=1}^m
U\!\left(\frac{t_k - t}{h_1}\right)
U\!\left(\frac{t_k - t}{h_1}\right)^\top
K\!\left(\frac{t_k - t}{h_1}\right)$,
with $U(u)=(1,u,\ldots,u^\ell)^\top$.

We will impose regularity conditions that are widely used in the local polynomial smoothing literature.
\begin{assumption}[LP1--LP3, {\citealp[p.~37]{Tsybakov}}]
\label{assump:LP} 
Assume the following regularity conditions for the local polynomial smoothing:
\begin{enumerate}
    \item There exists a real number \(\lambda_0 > 0\) and a positive integer \(M_0\) such that for all \(m \ge m_0\) and \(t \in [0,1]\), the smallest eigenvalue of \(B_{mt}\), denoted by \(\lambda_{\min}(B_{mt})\), satisfies $\lambda_{\min}(B_{mt}) \;\ge\; \lambda_0.$
    \item There exists a real number \(a_0 > 0\) such that for any interval \(I \subseteq [0,1]\) and all \(m \ge 1\), $
        \frac{1}{m}\sum_{k=1}^m \mathbf{1}\left\{\,t_k \in I\right\}\le a_0\,\max \bigl(|I|,\;1/m\bigr),
    $
    where $|I|$ denotes the length of interval \(I\).
    \item The kernel function \(K\) has compact support within \([-1,1]\). There exists a constant \(K_{\max} < \infty\) such that \(\lvert K(u)\rvert \le K_{\max}\) for all \(u \in \mathbb{R}\).
\end{enumerate}
\end{assumption}
\noindent
They ensure that the local design matrix, $B_{mt}$, is well behaved, the time points $\{t_k\}$ are not overly concentrated on small intervals, and the kernel function satisfies the usual boundedness and compact support conditions.

With the previous assumptions, we now present a pointwise error bound for $\hat{P}$.
\begin{theorem}\label{thm:double-smooth}
Under Assumptions~\ref{assump:smooth}, \ref{assump:lipschitz}, and \ref{assump:LP}, suppose $\rho_n\ge\sqrt{\tfrac{\log nm}{nm}}$. %
\ifreview
\revdel{Set the temporal and neighborhood bandwidths to be}
\begin{equation*}
\begin{aligned}
\text{\revdel{\mbox{$(h_1^\star,h_2^\star)\asymp$}}}
\begin{cases}
\text{\revdel{\mbox{$
\Bigl(
(\tfrac{\log nm}{\rho_n^2nm})^{1/(2\beta+1)},
\;\rho_n^{1/(2\beta+1)}
(\tfrac{\log nm}{nm})^{\beta/(2\beta+1)}
\Bigr)
$}}},
&
\text{\revdel{\mbox{if $n\rho_n\gtrsim
\sqrt{nm^{\beta/(\beta+1)}\log nm}$}}},
\\[6pt]
\text{\revdel{\mbox{$
\Bigl(
(\tfrac{\log nm}{\rho_n m})^{1/(2\beta+1)},
\;1/n
\Bigr)
$}}},
&
\text{\revdel{\mbox{if $n\rho_n\lesssim
\sqrt{nm^{\beta/(\beta+1)}\log nm}$}}},
\end{cases}
\end{aligned}
\end{equation*}
\revdel{and let $\hat P^{(t)}$ denote the corresponding two-stage smoothing estimator~\mbox{\eqref{eq:Phat}}.}
\fi%
\revadd{If $n\rho_n\gtrsim\sqrt{nm^{\beta/(\beta+1)}\log nm}$, set the temporal and neighborhood bandwidths to be
\begin{equation*}
\begin{aligned}
(h_1^\star,h_2^\star)\asymp
\Bigl(
(\tfrac{\log nm}{\rho_n^2nm})^{1/(2\beta+1)},
\;\rho_n^{1/(2\beta+1)}
(\tfrac{\log nm}{nm})^{\beta/(2\beta+1)}
\Bigr).
\end{aligned}
\end{equation*}
Otherwise, set $h_1^\star\asymp(\tfrac{\log nm}{\rho_n m})^{1/(2\beta+1)}$ and $\N_i^{(t)}=\{i\}$ for every $i$ and $t$. Let $\hat P^{(t)}$ denote the corresponding two-stage smoothing estimator~\eqref{eq:Phat}, which reduces to $\Tilde P^{(t)}$ in the second regime.} 
Then there exist constants $C,\tilde C>0$, depending only on the time-varying graphon and the kernel in \eqref{eq:local-poly}, such that for any $t\in[0,1]$ and for sufficiently large $n,m$,
\begin{equation*}
\begin{aligned}
\frac{1}{n}\,\|\hat P^{(t)}-P^{(t)}\|_{2,\infty}^2
\;\le\;
\begin{cases}
C\,\rho_n^{\frac{2\beta+2}{2\beta+1}}
\Bigl(\dfrac{\log nm}{nm}\Bigr)^{\!\frac{\beta}{2\beta+1}},
&
n\rho_n\gtrsim
\sqrt{nm^{\frac{\beta}{\beta+1}}\log nm},
\\[6pt]
C\,\rho_n^{\frac{2\beta+2}{2\beta+1}}
\Bigl(\dfrac{\log nm}{m}\Bigr)^{\!\frac{2\beta}{2\beta+1}},
&
\text{otherwise,}
\end{cases}
\end{aligned}
\end{equation*}
with probability at least $1-2(nm)^{-(\tilde C+\gamma)}$ for some $\gamma>0$.
\end{theorem}
Theorem~\ref{thm:double-smooth} establishes the error bounds for the two-stage smoothing estimator~\eqref{eq:Phat}, characterizing its behavior across two distinct asymptotic regimes. The first regime occurs when only a limited number of observation time points are available ($n \gtrsim m$), under which the sparsity condition $n\rho_n \gtrsim \sqrt{n m^{\beta/(\beta+1)}\log(nm)}$ holds. In this setting, the node-domain neighborhood smoothing step is crucial for achieving accurate estimation. Setting the bandwidth $h_2^\star \asymp \sqrt{\log(nm)/(nmh_1^\star)}$ yields non-degenerate neighborhoods, effectively pooling structural information across nodes. The second regime emerges when the number of time points substantially exceeds the network size ($n \ll m$). In this data-rich temporal regime, local polynomial smoothing alone is sufficient to recover the probability matrix. \revdel{Correspondingly, our theoretical framework adaptively sets $h_2^\star \asymp 1/n$, which essentially bypasses the neighborhood smoothing step.} \revadd{Correspondingly, in this regime, we set $\N_i^{(t)}=\{i\}$ for every $i$ and $t$, so that $\hat P^{(t)}=\Tilde P^{(t)}$ and the neighborhood smoothing step is omitted.}

\begin{corollary}\label{cor:double-smooth-sym}
\revadd{Under the conditions and estimator choices of Theorem~\ref{thm:double-smooth}, let $\hat P_{Sym}^{(t)}=(\hat P^{(t)}+(\hat P^{(t)})^\top)/2$. Then, for any $t\in[0,1]$ and for sufficiently large $n,m$,
\begin{equation*}
\begin{aligned}
\frac{1}{n^2}\,\|\hat P_{Sym}^{(t)}-P^{(t)}\|_F^2
\;\le\;
\begin{cases}
C\,\rho_n^{\frac{2\beta+2}{2\beta+1}}
\Bigl(\dfrac{\log nm}{nm}\Bigr)^{\!\frac{\beta}{2\beta+1}},
&
n\rho_n\gtrsim
\sqrt{nm^{\frac{\beta}{\beta+1}}\log nm},
\\[6pt]
C\,\rho_n^{\frac{2\beta+2}{2\beta+1}}
\Bigl(\dfrac{\log nm}{m}\Bigr)^{\!\frac{2\beta}{2\beta+1}},
&
\text{otherwise,}
\end{cases}
\end{aligned}
\end{equation*}
with probability at least $1-2(nm)^{-(\tilde C+\gamma)}$, where $C$, $\tilde C$, and $\gamma$ are as in Theorem~\ref{thm:double-smooth}.}
\end{corollary}
\revadd{Since $\|M\|_F\le\sqrt{n}\|M\|_{2,\infty}$ for any $M\in\mathbb{R}^{n\times n}$, Theorem~\ref{thm:double-smooth} immediately yields a Frobenius error bound. As in \cite{zhang2017estimating}, this Frobenius error bound continues to hold for the symmetrized estimator $\hat P_{Sym}^{(t)}$.}

To bridge our findings with the existing literature, we consider the special case where the graphon is static, $f(x,y,t)=f_0(x,y)$, and the number of time points $m$ is bounded. All snapshots then have the same probability matrix, and the temporal weights sum to one, so Step~1 introduces no temporal bias. We may therefore use an admissible $h_1\asymp1$ to combine the snapshots; because their number is bounded, this changes only constant factors in the rate as $n$ grows. With $h_2\asymp(\log n/n)^{1/2}$, the neighborhood smoothing step yields the rate $\rho_n(\log n/n)^{1/2}$ for the $L_{2,\infty}$ error at each fixed $t$, under $\rho_n\gtrsim(\log n/n)^{1/2}$. Symmetrization gives the same order for the Frobenius error. For the single-network setting, \cite{yang2025network} generalized the result of \cite{zhang2017estimating} and demonstrated that the standard neighborhood smoothing estimator achieves an error rate of $\rho_n(\log n/n)^{1/2}$, agreeing in order with this Frobenius bound. This comparison connects our time-varying framework with the classical static graphon estimation problem.

Finally, we state the theoretical guarantees for the three-stage refinement estimator \eqref{eq:third-stage}. The smoothness guaranteed by the additional smoothing step delivers a uniform error rate over $t$, preserving the error rates established in Theorem~\ref{thm:double-smooth}.

\begin{theorem}\label{thm:three-step}
Under the assumptions of Theorem~\ref{thm:double-smooth}, let $\bar P^{(t)}$ be the estimator defined in \eqref{eq:third-stage}. %
\ifreview
\revdel{Set the bandwidth $h_3$ as}
\begin{equation*}
\begin{aligned}
\text{\revdel{\mbox{$h_3^\star=$}}}
\begin{cases}
\text{\revdel{\mbox{$
\rho_n^{-2/(2\beta+1)}
\bigl(\tfrac{\log nm}{nm}\bigr)^{1/(2\beta+1)}
$}}},
&
\text{\revdel{\mbox{$
n\rho_n\gtrsim\sqrt{nm^{\beta/(\beta+1)}\log nm}
$}}},
\\[6pt]
\text{\revdel{\mbox{$
\rho_n^{-2/(2\beta+1)}
\bigl(\tfrac{\log nm}{m}\bigr)^{2/(2\beta+1)}
$}}},
&
\text{\revdel{\mbox{$
n\rho_n\lesssim\sqrt{nm^{\beta/(\beta+1)}\log nm}
$}}}.
\end{cases}
\end{aligned}
\end{equation*}
\fi%
\revadd{Use the same kernel and local polynomial order as in the first stage, and set $h_3=h_1^\star$, where
\begin{equation*}
\begin{aligned}
h_3^\star=h_1^\star\asymp
\begin{cases}
\bigl(\tfrac{\log nm}{\rho_n^2nm}\bigr)^{1/(2\beta+1)},
&
n\rho_n\gtrsim\sqrt{nm^{\beta/(\beta+1)}\log nm},
\\[6pt]
\bigl(\tfrac{\log nm}{\rho_n m}\bigr)^{1/(2\beta+1)},
&
n\rho_n\lesssim\sqrt{nm^{\beta/(\beta+1)}\log nm}.
\end{cases}
\end{aligned}
\end{equation*}
}
\revdel{There }\revadd{Then there} exist constants $C,\tilde C>0$ such that, \revadd{ for sufficiently large $n,m$,}
\begin{equation*}
\begin{aligned}
\sup_{t\in[0,1]}\frac{1}{n}\|\bar P^{(t)}-P^{(t)}\|_{2,\infty}^2
\;\le\;
\begin{cases}
C\rho_n^{\frac{2\beta+2}{2\beta+1}}
\bigl(\tfrac{\log nm}{nm}\bigr)^{\beta/(2\beta+1)},
&n\rho_n\gtrsim\sqrt{nm^{\beta/(\beta+1)}\log nm},
\\[6pt]
C\rho_n^{\frac{2\beta+2}{2\beta+1}}
\bigl(\tfrac{\log nm}{m}\bigr)^{2\beta/(2\beta+1)},
&n\rho_n\lesssim\sqrt{nm^{\beta/(\beta+1)}\log nm},
\end{cases}
\end{aligned}
\end{equation*}
with probability at least $1-2(nm)^{-(\tilde C+\gamma)}$ for some $\gamma>0$.
\end{theorem}
It is worth noting that, although temporal local polynomial smoothing inevitably encounters boundary effects due to asymmetric effective neighborhoods near the endpoints of $[0,1]$, Assumption~\ref{assump:LP} guarantees that the requisite regularity conditions hold uniformly. Consequently, the derived error rates remain theoretically valid for all $t \in [0,1]$. While the empirical estimation error may exhibit an increase near the temporal boundaries due to these one-sided smoothing neighborhoods, this boundary effect is strictly controlled and does not alter the asymptotic order of the bounds.

\section{Simulation Experiments} \label{sec:simulation}
In this section, we evaluate the empirical performance of the proposed method using extensive simulation studies. To provide a comprehensive assessment, we compare our approach against several benchmark methods applicable to the same time-varying data structure.

To thoroughly evaluate the generality and robustness of our method, we design a diverse set of generative scenarios that encompass various structural assumptions and temporal dynamics, detailed as follows.

\paragraph{Dynamic SBM-sine.} We first consider a dynamic stochastic block model (SBM) with $K=4$ communities. The within-block and between-block connection probabilities are governed by $a\bigl(1 + 0.5\sin(2\pi t)\bigr)$ and $b\bigl(1 + 0.5\sin(2\pi t)\bigr)$, respectively. We set the out-in ratio to $b/a = 0.5$, carefully calibrating the scaling parameters $a$ and $b$ to achieve a desired expected average node degree at the initial time point. Because all connection probabilities modulate according to a shared sine function trend, the block probability matrix at each time point is simply a scalar multiple of a fixed matrix. Consequently, it remains positive definite across all temporal snapshots. We refer to this baseline parametric setting as ``Dynamic SBM-sine.''

\paragraph{Dynamic SBM-NPD.} To evaluate the robustness of our method beyond the structurally restrictive setting of a shared global trend, we introduce a modified dynamic SBM characterized by heterogeneous temporal dynamics. Specifically, the four communities are partitioned into two distinct groups: for the first two communities, the within-block probabilities follow the inverted trend $a\bigl(1 - 0.5\sin(2\pi t)\bigr)$, while for the remaining two, they follow the original pattern $a\bigl(1 + 0.5\sin(2\pi t)\bigr)$. For the between-block connections, each community pair is randomly assigned one of these two temporal trajectories, ensuring that approximately half of the between-block probabilities evolve according to $b\bigl(1 + 0.5\sin(2\pi t)\bigr)$ and the other half follow $b\bigl(1 - 0.5\sin(2\pi t)\bigr)$. Intuitively, this targeted construction breaks the global scaling structure, producing network snapshots where the block probability matrix becomes non-positive definite. We designate this generative model as ``Dynamic SBM-NPD.''



\paragraph{Dynamic RDPG.} We next consider the dynamic RDPG introduced by \cite{macdonald2025latent}, wherein the latent node positions evolve smoothly over time. To ensure a fair and direct comparison with the FASE method, we adopt the exact simulation specifications detailed in \cite{macdonald2025latent}, using the R package \texttt{fase} \citep{fase}.

\paragraph{Dynamic latent distance model.} We also evaluate our method under the latent distance model \citep{hoff2002latent, sewell2015latent}. Here, the underlying latent trajectories are generated identically to those in the dynamic RDPG setting. However, in contrast to the RDPG, the latent distance formulation relies on a non-linear logit link function, $P_{ij}^{(t)} = \text{logit}^{-1}(\alpha -\|z^{(t)}_i - z^{(t)}_j\|_2)$ where $z_i$ denotes the latent position (e.g., generated via the \texttt{fase} package). As a result, a probability matrix $P^{(t)}$ is not low-rank at any given time point.

Across all evaluated scenarios, we simulate time-varying networks comprising $n=600$ nodes, observed over $m=100$ equally spaced time points $t \in \{1/m, 2/m, \dots, 1\}$. To ensure rigorous empirical comparability, the generative parameters within each model are systematically calibrated to produce an expected initial average node degree of approximately 50.

To provide a comprehensive assessment of our proposed approach, we include several benchmark methods in our empirical evaluation:

\begin{itemize}
    \item \textbf{Reversed}: This is the variant of our multi-stage smoothing procedure discussed in Remark~\ref{rem:smoothing-order}. It reverses the order of the two smoothing operators by applying node-domain neighborhood smoothing prior to temporal smoothing. 
    \item \textbf{Independent Neighborhood Smoothing}: This method applies standard neighborhood smoothing  \citep{zhang2017estimating} to each temporal snapshot independently, establishing a baseline that completely ignores temporal information.
    \item \textbf{Pooled Neighborhood Smoothing} \citep{zhao2019change}: This approach applies neighborhood smoothing directly to the time-averaged network, thereby yielding a single time-invariant probability estimate. It represents the extreme case of imposing a constant model over time. 
    \item \textbf{The FASE method} \citep{macdonald2025latent}: This approach uses the penalized estimation framework designed for the dynamic RDPG. It can be seen as an estimator for time-varying low-rank, positive semi-definite probability matrices. 
    \item \textbf{Dynamic Spectral Clustering} \citep{pensky2019spectral}: This approach first smooths the adjacency matrices using kernel over time to estimate the time-varying edge probability matrices, and then applies spectral clustering to recover the community structure. Based on the estimated communities, we further construct a block-wise estimator by averaging the entries of the estimated probability matrix within each pair of communities, yielding a SBM–type approximation. 
\end{itemize}

For the reversed-order baseline, the tuning parameter is determined using an oracle criterion that explicitly minimizes the estimation error, thereby establishing an idealized upper bound on its performance. For all remaining benchmark methods, the tuning parameters are systematically selected via data-driven cross-validation, strictly adhering to the protocols recommended in their respective original publications. Finally, to implement the temporal local polynomial smoothing within our proposed framework, we employ a compactly supported tricube kernel.

In terms of implementation, we closely follow the publicly available codes or recommended procedures from the original works whenever possible. In particular, the implementation of the dynamic spectral clustering method is based on the code provided by \cite{pensky2019spectral}. The neighborhood smoothing procedure is implemented following the approach in the \texttt{randnet} package \citep{randnet}. For latent position generation, we utilize the \texttt{fase} package \citep{fase}.

The performance is evaluated at each time point using relative errors under the Frobenius norm and the $L_{2,\infty}$ norm. Let $P^{(t)}$ denote the true probability matrix at time $t$ and \revdel{\mbox{$\hat P^{(t)}$} }\revadd{$\hat P_{Sym}^{(t)}$} its estimator. The relative errors at time $t$ are defined as
\begin{align*}
\mathrm{Err}_F^{(t)}
&=\frac{\|
\ifreview
\text{\revdel{\mbox{$\hat P^{(t)}$}}}\;
\fi
\revadd{\hat P_{Sym}^{(t)}}-P^{(t)}\|_F}{\|P^{(t)}\|_F},
\quad
\mathrm{Err}_{2,\infty}^{(t)}=
\frac{\|
\ifreview
\text{\revdel{\mbox{$\hat P^{(t)}$}}}\;
\fi
\revadd{\hat P_{Sym}^{(t)}}-P^{(t)}\|_{2,\infty}}{\|P^{(t)}\|_{2,\infty}}.
\end{align*}
Under each experimental configuration, we generate 100 independent replicates and report the average errors across replicates for each fixed time point $t$.

\begin{figure}[htbp]
    \centering
    \begin{subfigure}{\textwidth}
        \centering
        \includegraphics[width=\textwidth]{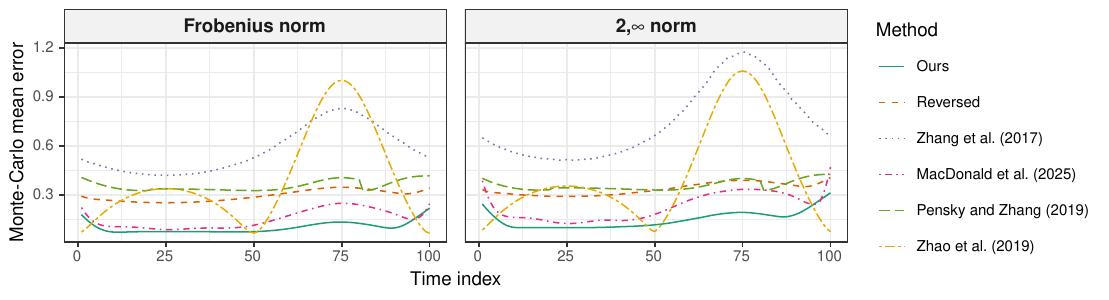}
        \vspace{-5mm}
        \caption{Dynamic SBM-sine}
    \end{subfigure}\vspace{0.5em}

    \begin{subfigure}{\textwidth}
        \centering
        \includegraphics[width=\textwidth]{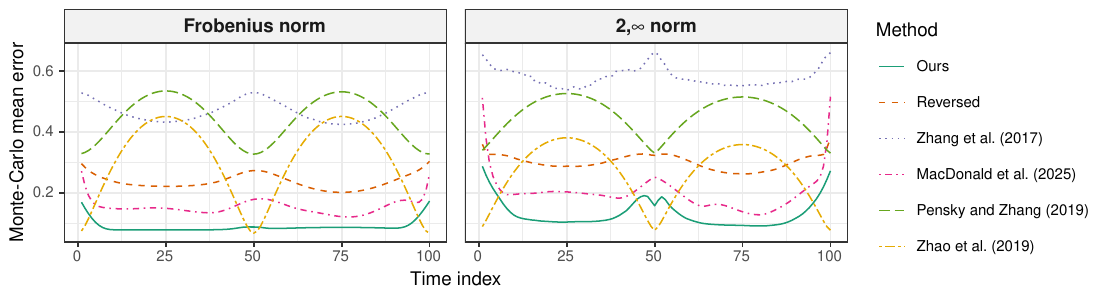}
        \vspace{-5mm}
        \caption{Dynamic SBM-NPD}
    \end{subfigure}\vspace{0.5em}

    \begin{subfigure}{\textwidth}
        \centering
        \includegraphics[width=\textwidth]{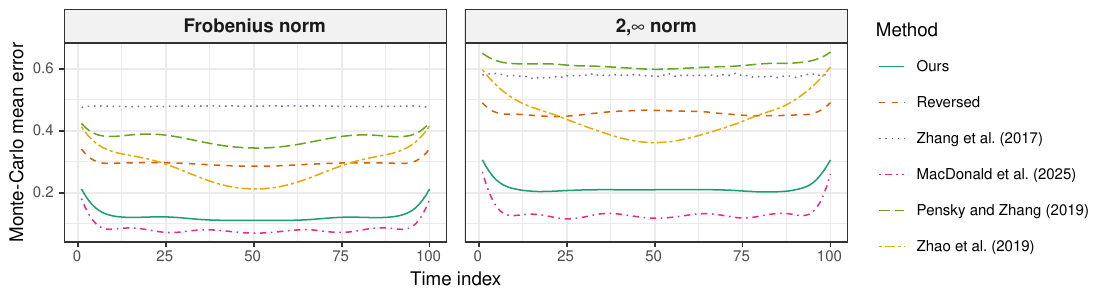}
        \vspace{-5mm}
        \caption{Dynamic RDPG}
    \end{subfigure}\vspace{0.5em}

    \begin{subfigure}{\textwidth}
        \centering
        \includegraphics[width=\textwidth]{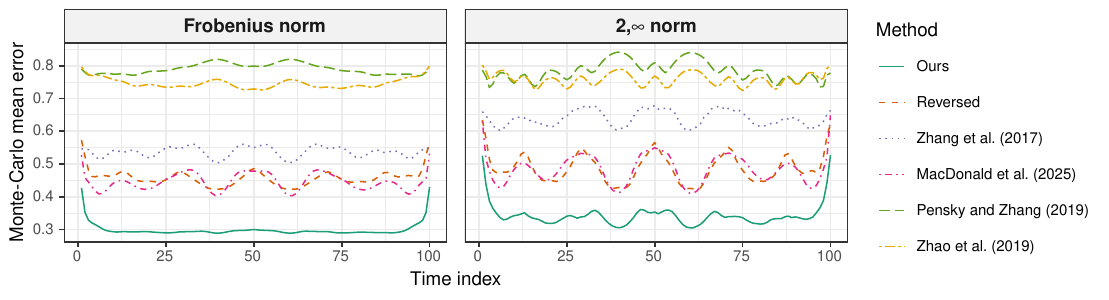}
        \vspace{-5mm}
        \caption{Dynamic latent distance model}
    \end{subfigure}
      \caption{Performance of six estimators under different data generating models. Each curve represents the Monte Carlo mean error over 100 replications.}
    \label{fig:error_curves_all}
\end{figure}

Figure \ref{fig:error_curves_all} reports the mean relative errors over time across the four data-generating mechanisms. The performance comparison results under the Frobenius norm and $L_{2,\infty}$ norm are consistent with each other. For our method, a moderate increase in error is observed near the temporal boundaries. This boundary effect is expected for temporal local polynomial smoothing over a finite time horizon. Apart from these boundary regions, the error curves remain comparatively stable over the interior time interval.  Also, our proposed method consistently outperforms the ``Reversed'' variant (with oracle tuning). This uniform advantage convincingly demonstrates the importance of the smoothing order in the current estimation problem, as discussed in Remark~\ref{rem:smoothing-order}.

The independent neighborhood smoothing method of \cite{zhang2017estimating}, which completely ignores the temporal relation, fails to achieve reasonable performance in the current context.  The pooled neighborhood smoothing of \cite{zhao2019change}, which assumes identical models over time, can give a reasonable estimate only at a few time points when the pattern happens to be close to the over-time average, but generally fails to deliver an adequate accuracy. The dynamic spectral clustering of \cite{pensky2019spectral} is not effective either. This could be due to the instability of spectral clustering when combined with the smoothing operator. Another potential reason for this, under the SBM, could be the signal cancellation, as investigated by \cite{lei2023bias}.

Among all evaluated benchmark methods, the only approach that achieves comparable empirical performance to ours in specific settings is the FASE method of \cite{macdonald2025latent}. Specifically, FASE performs comparably to our method under the SBM-sine configuration and exhibits superior accuracy under the dynamic RDPG model. These empirical patterns align with theoretical expectations: FASE is explicitly designed for the dynamic RDPG framework, which structurally encompasses the dynamic SBM with a positive semi-definite probability matrix (such as the SBM-sine) as a special case. Under such correct model specification, FASE leverages its restrictive parametric design to achieve high statistical efficiency. Nevertheless, our nonparametric approach remains highly competitive even within the dynamic RDPG regime, sacrificing minimal accuracy in exchange for its broader generality.

In contrast, under both the SBM-NPD and the dynamic latent distance models, FASE exhibits a marked degradation in performance and is strictly inferior to our proposed estimator. This divergence highlights the vulnerability of rigid structural assumptions: the SBM-NPD setting does not consistently yield positive semi-definite probability matrices over time, and the dynamic latent distance model inherently falls outside the RDPG class. Thanks to the structural generality of the time-varying graphon framework, our method seamlessly adapts to these complex temporal dynamics, maintaining robust estimation accuracy where FASE falters. Overall, our proposed estimator demonstrates impressive adaptivity and principled robustness across all generative settings under evaluation.

\section{Data Application: U.S. Senate Cosponsorship} \label{sec:data}
To demonstrate our proposed methodology, we analyze the U.S. Senate cosponsorship network from the 93rd to the 118th Congress (1973--2024). Our analysis integrates two complementary data sources due to differences in historical coverage. For earlier Congresses (93rd--108th, 1973--2004), we used the curated cosponsorship dataset available through Harvard Dataverse \citep{fowler2006connecting}, which provides structured bill-level sponsorship and cosponsorship information. For more recent Congresses (109th onward), we retrieved the official legislative bill status and cosponsorship records directly from the Library of Congress bulk data repository (\url{https://www.govinfo.gov/bulkdata/BILLSTATUS}), as the Dataverse dataset does not extend beyond 2004. We aggregated individual senators to the level of their home states. For each calendar year, we constructed an unweighted, undirected network among the $50$ U.S. states: an edge is established between two states (set to $1$) if the total volume of cosponsored bills between their respective senators strictly exceeds the yearly mean, defined as the average of these pairwise cosponsorship totals across all state pairs in that year, and to $0$ otherwise. This aggregation yields a discrete sequence of $m = 52$ time-varying network snapshots, which we use to examine the temporal evolution of cross-state legislative collaboration.

Our primary objective is to characterize how the cross-state cosponsorship structure has evolved over time, and to identify groups of states whose connection patterns follow common temporal dynamics. To this end, we directly apply the proposed two-stage smoothing estimator \eqref{eq:Phat} to the observed sequence $\{A^{(t_k)}\}_{k=1}^{52}$, yielding a sequence of smoothed probability matrices 
\revadd{$\{\hat{P}_{Sym}^{(t_k)}\}_{k=1}^{52}$}. This produces an estimated trajectory for each state pair. The tuning parameters $(\ell, h_1, h_2)$ are selected by the leave-one-time-out cross-validation procedure described in Algorithm~\ref{alg:cv}. \revadd{Approximately one-third of Senate seats are up for election every two years, and the resulting continuity in membership motivates a smooth temporal model for state-level cosponsorship patterns. In Appendix~\ref{app:smoothness-check}, we compare estimates fitted separately to the years before and after each annual boundary and find no dominant break.} Building on the estimated probability matrices, we summarize each state by its full connectivity trajectory across the entire observation horizon. For each state $i\in\{1,\ldots,50\}$, let $v_i\in\mathbb{R}^{2600}$ denote the trajectory vector obtained by stacking the $i$-th row of \revdel{\mbox{$\hat P^{(t)}$} }\revadd{$\hat P_{Sym}^{(t)}$} across all $52$ years, so that $v_i$ records state $i$'s estimated connection probabilities over the observation period. We then compute pairwise Pearson correlations among $\{v_i\}_{i=1}^{50}$ to obtain a $50 \times 50$ similarity matrix $S$, rescale its entries to $[0,1]$ via $\tilde{S}_{ij} = (S_{ij} + 1)/2$, and define the dissimilarity matrix as $D_{ij} = 1 - \tilde{S}_{ij}$.

To identify groups of states with similar connectivity trajectories, we apply Ward hierarchical clustering \citep{murtagh2014ward} to $D$. We select the dendrogram cut by maximizing the average Silhouette criterion over $K \in \{3,4,5,6\}$. The upper bound $K \le 6$ is imposed to keep the number of clusters small enough for interpretation. The optimal cut is achieved at $K = 4$, and Table~\ref{tab:clusters} reports the resulting cluster memberships. To examine the temporal pattern within each cluster, we compute the mean within-cluster affiliation
\begin{equation*}
\bar{P}_{aa}^{(t)} \;=\; \frac{1}{|C_a|\,(|C_a|-1)} \sum_{\substack{i, j \in C_a \\ i \neq j}} \revadd{\hat{P}_{Sym, ij}^{(t)}},
\end{equation*}
where $C_a$ denotes the set of states in the $a$-th cluster ($a = 1, 2, 3, 4$). This yields a single scalar per year per cluster. Figure~\ref{fig:clusters} displays the four resulting trajectories over $1973$--$2024$. The four clusters exhibit clearly distinct temporal patterns, summarized below.

\begin{table}[t]
\centering
\begin{tabular}{|c|p{0.78\textwidth}|}
\hline
\textbf{Cluster} & \textbf{States} \\
\hline
1 & AL, AZ, FL, IA, KY, MS, MT, NV, NC, ND, SC, SD, TN, UT, WV \\
\hline
2 & AK, ID, IN, KS, MO, NH, OK, TX, WY \\
\hline
3 & AR, GA, LA, NE \\
\hline
4 & CA, CO, CT, DE, HI, IL, ME, MD, MA, MI, MN, NJ, NM, NY, OH,
OR, PA, RI, VT, VA, WA, WI \\
\hline
\end{tabular}
\caption{$K=4$ cluster memberships for the 50 U.S. states.}
\label{tab:clusters}
\end{table}

\begin{figure}[htbp]
    \centering
    \begin{subfigure}{0.8\linewidth}
        \centering
        \includegraphics[width=\linewidth]{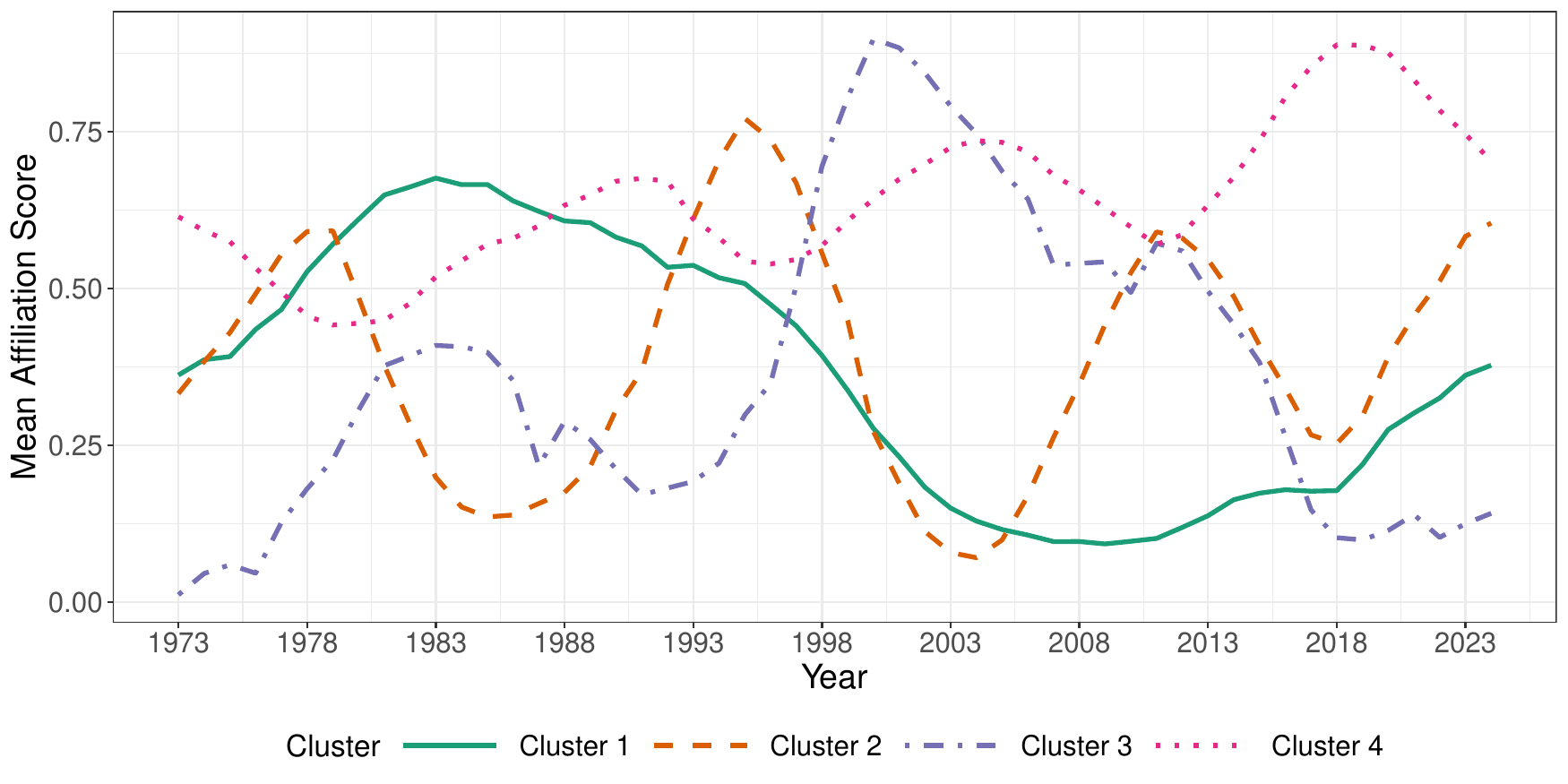}
        \caption{Cluster-wise mean trajectories over time ($K=4$).}
    \end{subfigure}

    \vspace{0.75em}

    \begin{subfigure}{0.45\linewidth}
        \centering
        \includegraphics[width=\linewidth]{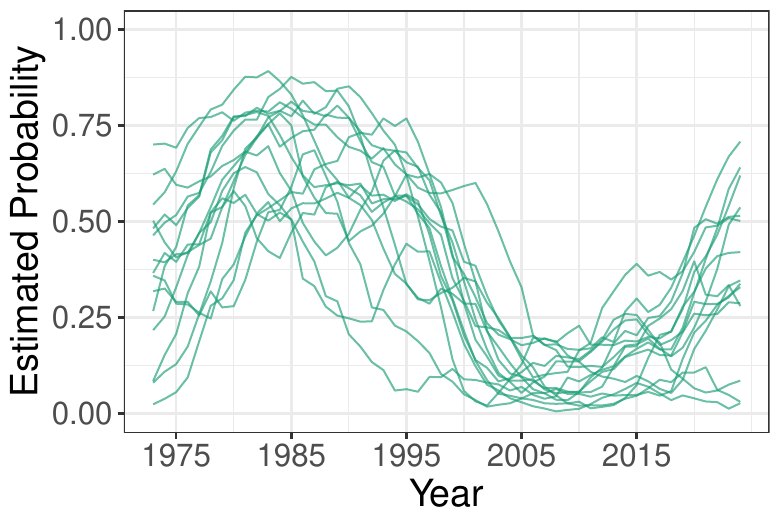}
        \caption{Cluster 1}
    \end{subfigure}
    \hfill
    \begin{subfigure}{0.45\linewidth}
        \centering
        \includegraphics[width=\linewidth]{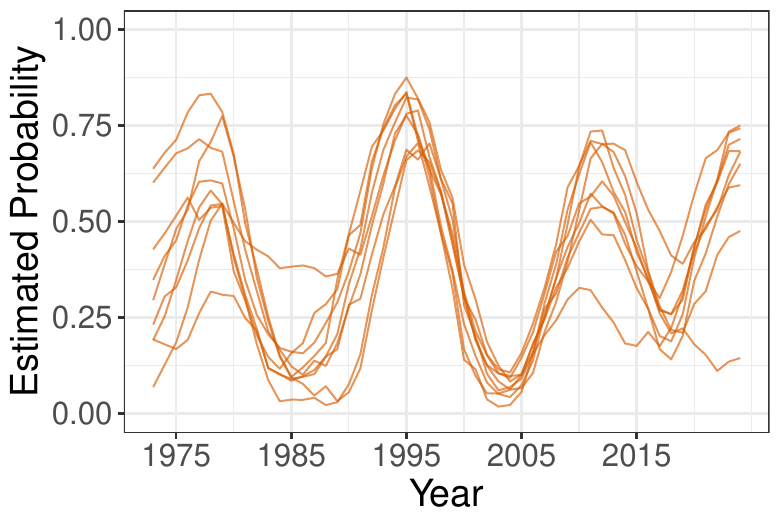}
        \caption{Cluster 2}
    \end{subfigure}

    \vspace{0.5em}

    \begin{subfigure}{0.45\linewidth}
        \centering
        \includegraphics[width=\linewidth]{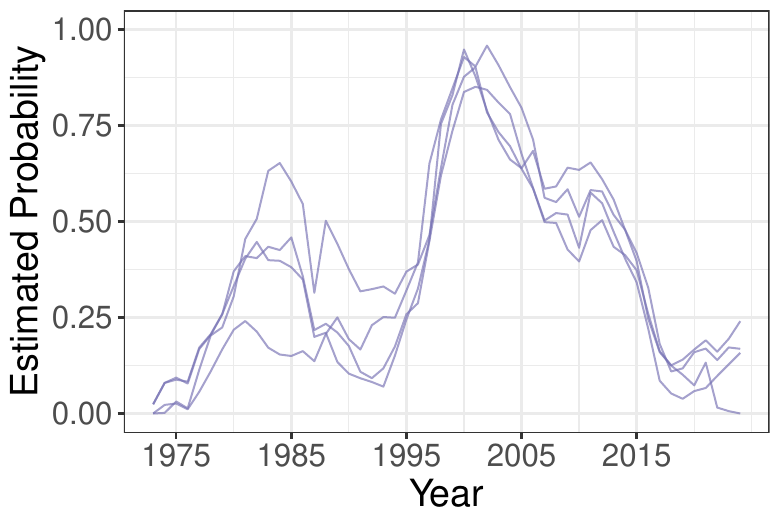}
        \caption{Cluster 3}
    \end{subfigure}
    \hfill
    \begin{subfigure}{0.45\linewidth}
        \centering
        \includegraphics[width=\linewidth]{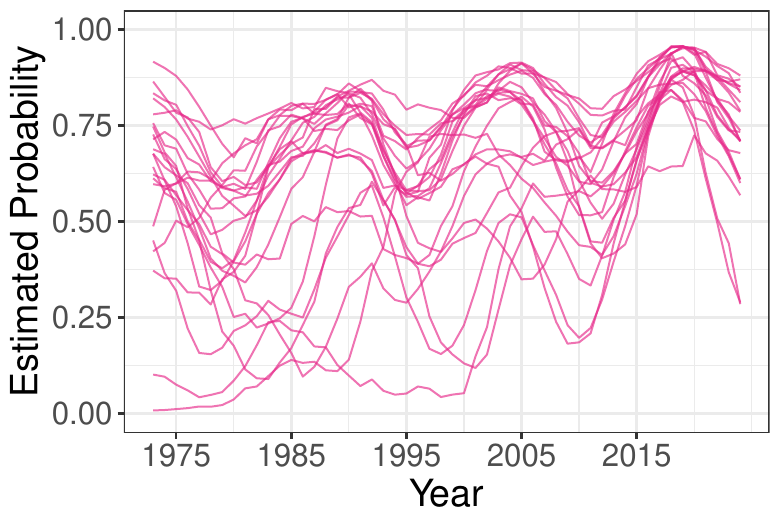}
        \caption{Cluster 4}
    \end{subfigure}
    \caption{
    Cluster-level mean trajectories. 
    The top panel displays the cluster-wise mean trajectories over time for the $K=4$ clusters identified by hierarchical clustering. 
    The lower panels show individual state trajectories stratified by cluster.
    }
    \label{fig:clusters}
\end{figure}

Cluster 1 starts at moderate levels in the mid-1970s, rises through the late 1970s and early 1980s, and remains relatively elevated through the early 1990s before declining sharply around the turn of the century. It then stays at low levels for much of the 2000s and 2010s, with a modest rebound in the most recent years. This suggests that the states in this cluster shared a strong but historically localized pattern of cosponsorship behavior that weakened substantially in the more recent period, consistent with the erosion of an older conservative legislative alignment. Cluster 2 displays the most strikingly cyclic pattern, with three pronounced peaks around 1980, 1995, and the early 2020s, separated by periods of substantially lower affiliation. The mid-1990s peak coincides with the Republican Revolution following the 1994 elections, but the recurrence of similarly sharp peaks suggests that the states in this cluster align closely only during episodic legislative pushes rather than maintaining a stable common agenda. Cluster 3 is also non-monotone, with a modest rise in the mid-1980s, a sharp peak around the early 2000s, and another smaller increase in the early 2010s. Given its small size and predominantly Southern composition, these repeated increases are consistent with localized regional realignment rather than a stable long-run bloc. Cluster 4, the largest cluster, contains a broad collection of mostly coastal and industrial states and exhibits a consistently wide range of affiliation values throughout the period. Its trajectories are therefore best read as a broad baseline of shared legislative connectivity rather than as evidence of a sharply defined political coalition.

\begin{figure}[htbp]
    \centering
\includegraphics[width=\linewidth]{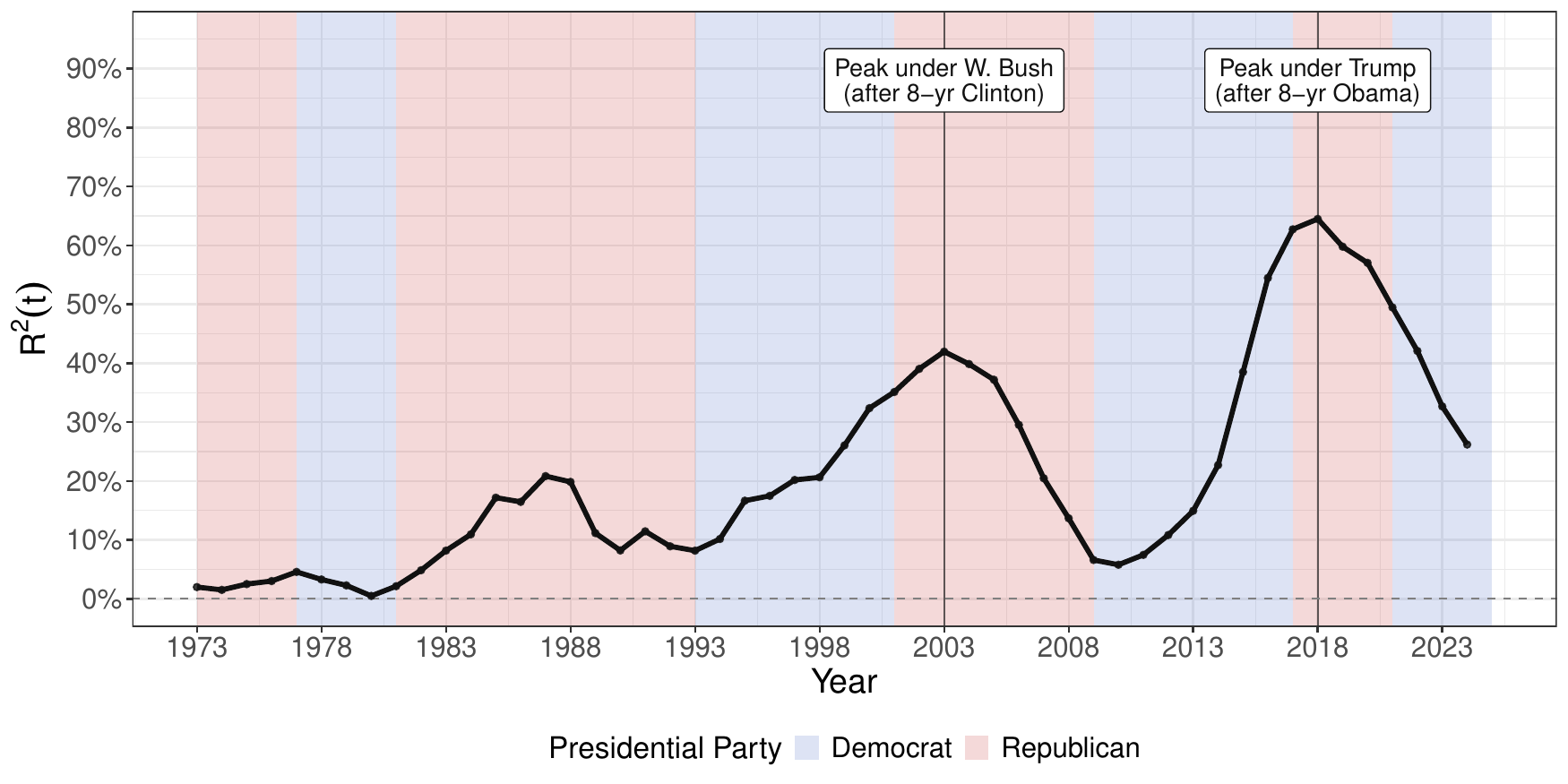}
    \caption{Group-based polarization score $R^{2}(t)$ for the smoothed cosponsorship network, 1973--2024. Background shading indicates the presidential party (blue $=$ Democrat, red $=$ Republican); vertical lines mark the most pronounced peaks of $R^{2}(t)$.}
    \label{fig:polarization}
\end{figure}

As an additional analysis, we examine temporal polarization in the cosponsorship network using a group-based polarization score motivated by \citet{mehlhaff2024group}. In each year $t$, every state is classified into one of three party-status categories---DD (both senators Democratic), RR (both Republican), or Mixed---and every state pair $(i,j)$ is correspondingly assigned to a pair-type group $g_t(i,j)$ based on the statuses of $i$ and $j$. Because senatorial composition changes through elections, the partition $g_t$ used to compute $R^{2}(t)$ is allowed to vary with $t$. Letting $\bar{P}^{(t)} = \binom{n}{2}^{-1} \sum_{i<j} \revadd{\hat{P}^{(t)}_{Sym,ij}}$ denote the global mean across pairs and $\bar{P}^{(t)}_{g}$ the within-group mean for pair-type group $g$, we set
\begin{equation*}
\mathrm{SS}_{\mathrm{total}}(t) = \sum_{i<j}\bigl(\revadd{\hat{P}^{(t)}_{Sym, ij}} - \bar{P}^{(t)}\bigr)^{2}, \qquad
\mathrm{SS}_{\mathrm{between}}(t) = \sum_{i<j}\bigl(\bar{P}^{(t)}_{g_t(i,j)} - \bar{P}^{(t)}\bigr)^{2},
\end{equation*}
and define
\begin{equation*}
R^{2}(t) \;=\; \frac{\mathrm{SS}_{\mathrm{between}}(t)}{\mathrm{SS}_{\mathrm{total}}(t)}.
\end{equation*}
A larger $R^{2}(t)$ indicates that pairwise cosponsorship is more strongly segregated along party-composition lines, which we interpret as a year-specific measure of party-based polarization in the network.

As shown in Figure~\ref{fig:polarization}, the polarization score $R^{2}(t)$ exhibits two prominent rise-and-peak cycles aligned with the presidential calendar. In the first cycle, $R^{2}(t)$ rises steadily through the eight years of the Clinton administration (1993--2000) and reaches a peak of about $42\%$ under George W.\ Bush in 2003. In the second cycle, the score rises sharply through the Obama years (2009--2016) and attains its highest value of roughly $64\%$ under Trump in 2018. The recurrence of this pattern---steady accumulation under a Democratic White House followed by a peak partway into the subsequent Republican presidency---suggests that party-based separation in cosponsorship is most strongly amplified after a Democratic-to-Republican turnover.

\section{Discussion} \label{sec:discussion}
In this paper, we introduced a nonparametric multi-stage smoothing estimator for time-varying networks and established finite-sample error bounds under mild temporal and node-level structural assumptions. Our theoretical analysis precisely characterizes the dependence of the estimation error on the network size $n$, the number of observation time points $m$, and the sparsity factor $\rho_n$, yielding both pointwise and uniform convergence guarantees. Empirically, the proposed estimator demonstrates highly competitive performance across a diverse set of generative models, consistently outperforming several existing spectral and smoothing-based alternatives. 

A primary advantage of our method lies in its structural generality: the time-varying graphon framework ensures broad applicability to a wide range of complex network topologies, effectively bypassing the restrictive parametric constraints of competing benchmark methods. However, a notable limitation of the current framework is the strict assumption of continuous temporal smoothness. Accommodating smoothly evolving temporal trajectories interrupted by occasional abrupt structural change points represents a natural and practically relevant generalization for future research. Furthermore, another promising extension would be to replace the node-domain neighborhood smoothing operator with more general aggregation approaches, such as the ensemble and mixing strategies proposed by \cite{ghasemian2020stacking} and \cite{li2024network}. Integrating these principled aggregation methods could potentially render the estimation procedure even more robust and adaptive, particularly in the notoriously challenging regime of sparse networks.

\bibliographystyle{abbrvnat}
\bibliography{ref}{}

\begin{appendix}
\section{Appendix: Proof}
\subsection{Proof of Theorem~\ref{thm:double-smooth}}\label{appenix:proof_twostage}
\label{sec:appendix}
In this appendix, we collect the assumptions, lemmas, and their proofs, as well as the proof of Theorem 1. For each pair of nodes $(i,j)$, recall that in Step 1 (see \eqref{eq:local-poly}) we estimate the edge probability over time using a local polynomial order, $\ell$. In particular,
\begin{align*}
\tilde P^{(t)}_{ij}
=
\sum_{k=1}^m
w_k(t;h_1)\,
A^{(t_k)}_{ij},
\end{align*}
where the weights $w_k(t;h_1)$ are the equivalent kernel weights of the local polynomial order $\ell$. These weights admit the representation
\begin{align*}
w_k(t;h_1) = \frac{1}{m h_1}
U^\top(0)\,
B_{mt}^{-1}\,
U\!\left(\frac{t_k - t}{h_1}\right)
K\!\left(\frac{t_k - t}{h_1}\right),
\end{align*}
where $U(u) = (1,u,\ldots,u^\ell)^\top,$
and
$B_{mt}=\frac{1}{m h_1}
\sum_{k=1}^m
U\!\left(\frac{t_k - t}{h_1}\right)
U\!\left(\frac{t_k - t}{h_1}\right)^\top
K\!\left(\frac{t_k - t}{h_1}\right).
$
To guarantee the stability of the local polynomial estimator, we impose the following regularity conditions.

\begin{lemma}[Proposition~1.12 (p.~36) and Lemma~1.3 (p.~38) in {\normalfont\citealp{Tsybakov}}]\label{lem:weights}
Under Assumption 3, for all \(m \ge m_0\), \(h \ge \tfrac{1}{2m}\), and \(t \in [0,1]\), the local polynomial weights \(w_k(t; h)\) satisfy:
\begin{enumerate}
    \item \(\displaystyle \sup_{m,t} |w_k(t;h)| \;\le\; \frac{B_0}{m\,h},\)
    \item \(\displaystyle \sum_{k=1}^m |w_k(t;h)| \;\le\; B_0,\)
    \item \(w_k(t;h) = 0\) if \(\lvert t_k - t\rvert > h,\)
    \item $\sum_{k=1}^m w_k(t;h) = 1$
    \item $\sum_{k=1}^m (t_k - t)^p \ w_k(t;h) = 0$ for $p = 1,2,...,\ell$
\end{enumerate}
where \(B_0 = \max\bigl\{2K_{\max}/\lambda_0,\;4K_{\max}\,a_0/\lambda_0\bigr\}\).
\end{lemma}

\bigskip

\begin{lemma}
    Under Assumptions 1 and 3, for any $i,j$ such that $i \neq j$ and $t \in [0,1]$,
    \begin{align}
    \label{lem2:bias}
    &|{\rm E} \Tilde P_{ij}^{(t)} - P_{ij}^{(t)}|=|b_{ij}(t)| \leq B_1 \rho_n h_1^\beta\\
    \label{lem2:var}
    &{\rm Var}(\Tilde P_{ij}^{(t)}) = \sigma^2_{ij}(t) \leq  B_2\rho_n(mh_1)^{-1}
    \end{align}
\end{lemma}

\begin{proof}
Firstly, the bias of a kernel estimator, (\ref{lem2:bias}), is bounded by using Taylor expansion.
\begin{equation}
\label{eq:bias1}
\begin{aligned}
    |b_{ij}(t_0)| &= |{\rm E} \Tilde P_{ij}^{(t_0)} - P_{ij}^{(t_0)}| = \Big|\sum_{k=1}^{m}w_k(t_0)\big( 
    P_{ij}^{(t_k)} - P_{ij}^{(t_0)}
    \big)\Big|\\[8pt]
    &\leq \Big|\sum_{k=1}^m w_k(t_0) \ \rho_n \ \frac{f^{(\ell)}\big(t_0 +\tau_0(t_k-t_0)\big) -f^{(\ell)}(t_0)}{l!} (t_k-t_0)^{\ell}\Big|\\[8pt]
    &\leq  \sum_{k=1}^m  |w_k(t_0)| \ \rho_n \ \frac{\ L_1|t_k-t_0|^{\beta}}{\ell !}\\[8pt]
    &=  \sum_{k=1}^m  \frac{\  \rho_n L_1|t_k-t_0|^{\beta}}{\ell!}|w_k(t_0)| I(|t_k-t_0| \leq h_1)\\[8pt]
    &\leq \sum_{k=1}^{m}\frac{\rho_n L_1\ h_1^{\beta}}{\ell !}|w_k(t_0)| \leq B_1 \rho_n h_1^{\beta}
\end{aligned}
\end{equation}
where $0<\tau_0<1$ and $B_1 = L_1B_0/\ell!$. From the first line to the second line, we used the fourth and fifth results in Lemma~\ref{lem:weights}. 
Next, the transition from the second line to the third line used Assumption~\ref{assump:smooth}.
Finally, the last line follows from the second and third results in Lemma~\ref{lem:weights}.

For the variance part, (\ref{lem2:var}), 
\begin{align*}
    {\rm Var}(\Tilde P_{ij}^{(t_0)}) = \sigma^2_{ij}(t_0)&= {\rm E}\Big[\sum_{k=1}^m w_k(t_0)\big(A_{ij}^{(t_k)} - P_{ij}^{(t_k)}\big)\Big]^2\\
    &= \sum_{k=1}^m w_k^2(t_0) {\rm Var}(A_{ij}^{({t_k})})
    \\&\leq \rho_n \sup_{m,t}|w_k(t;h)| \sum_{k=1}^m |w_k(t;h)|\\
     &\leq B_2 \rho_n / mh_1
\end{align*}
where $B_2 = B_0^2$.
From the Definition 1, we have
\({\rm Var}\!\bigl(A_{ij}^{(t)}\bigr)
= P_{ij}^{(t)} \bigl(1 - P_{ij}^{(t)}\bigr)
= \rho_n \, f\bigl(\xi_i, \xi_j, t\bigr) 
\bigl(1 - \rho_n \, f\bigl(\xi_i, \xi_j, t\bigr)\bigr) \leq \rho_n.\) We used the first and the second results in Lemma~1 for the transition to the last inequality. This completes the proof.
\end{proof}

Denote $I_j = [x_{j-1}, x_j)$ for $1 \leq j < J$ and $I_J =[x_{J-1},1]$ for the intervals in Assumption~\ref{assump:lipschitz}, and denote $\delta = \min_{1\leq j \leq J} |I_j|$. For any $\xi \in [0,1]$, let $I(\xi)$ denote the $I_j$ that contains $\xi$. Let $S_i(\Delta) = [\xi_i - \Delta, \xi_i + \Delta] \cap I(\xi_i)$ denote the ball of $\xi_i$ in which $f(x,y,t)$ is Lipschitz in $x \in S_i(\Delta)$ for any fixed $y$.

\begin{lemma} \label{lem3}
For arbitrary global constants $C_1, C_2 > 0$, define \(\Delta_{n,m} = (C_1+C_2)\sqrt{\tfrac{\log nm}{nmh_1}}\).
If $n,m$ are large enough so that \(\sqrt{\tfrac{n}{mh_1\log nm}} > \tfrac{6\,C_1 + 7\,C_2}{C_2^2}(1 + \gamma)\) for some $\gamma > 0$ and $\Delta_{n,m} \,<\, \delta/2$, then there exists $\Tilde C_1 > 0$ such that
\[
    {\rm Pr}\Biggl( 
        \min_{i} \Bigl|\frac{|\{\;i' \neq i : \xi_{i'} \in S_i(\Delta_{n,m})\}|}{n-1}\Bigr|
        \;\ge\; 
        C_1\sqrt{\frac{\log nm}{nmh_1}}
    \Biggr)
    \;\;\ge\;\; 
    1 \;-\; 2\,(nm)^{-(\tilde C_1 + \gamma)}.
\]

\begin{proof}
To control the neighborhood size $|S_i(\Delta_{n,m})|$, note that the latent positions $\xi_i$ lie in disjoint intervals $I_k$. Since $\Delta_{n,m} < \delta/2$, for any $i$ we have that either 
$[\xi_i - \Delta_{n,m},\, \xi_i] \subset I(\xi_i)$ or $[\xi_i,\, \xi_i + \Delta_{n,m}] \subset I(\xi_i)$. 
If this were not the case, then the interval $I(\xi_i)$ would be contained in 
$[\xi_i - \Delta_{n,m},\, \xi_i + \Delta_{n,m}]$, implying $|I(\xi_i)| \le 2\Delta_{n,m}$, 
which contradicts the condition $\Delta_{n,m} < \delta$. Therefore, for all $i$, $\Delta_{n,m} \;\le\; |S_i(\Delta_{n,m})| \;\le\; 2\Delta_{n,m}$. 

\noindent
By Bernstein's inequality, for each $i$ we have
\begin{equation*}
\begin{aligned}
&{\rm Pr}\left(
    \left|
    \frac{|\{i'\neq i:\xi_{i'}\in S_i(\Delta_{n,m})\}|}{n-1}
    -|S_i(\Delta_{n,m})|
    \right|
    \ge \epsilon
\right)  \,\le \,
2\exp\left(-\frac{(n-1)\epsilon^2}{2(2\Delta_{n,m}+\epsilon/3)} \right).
\end{aligned}
\end{equation*}
Set $\epsilon=C_2\sqrt{\log nm/(nmh_1)}$.
For $n\ge 6$, we have $(n-1)/2\ge 5n/12$.
Substituting
$\Delta_{n,m}=(C_1+C_2)\sqrt{\log nm/(nmh_1)}$
therefore gives
\begin{equation*}
\begin{aligned}
\frac{(n-1)\epsilon^2}
     {2(2\Delta_{n,m}+\epsilon/3)}
\,\ge\, \frac{5C_2^2}{4(6C_1+7C_2)}
\sqrt{\frac{n}{mh_1\log nm}}\log nm  \,>\, \frac54(1+\gamma)\log nm,
\end{aligned}
\end{equation*}
where the last inequality uses the condition in the statement of the lemma.
Hence, a union bound over $i=1,\ldots,n$ gives
\begin{equation}
\begin{aligned}
&{\rm Pr}\left(
    \max_i
    \left|
    \frac{|\{i'\neq i:\xi_{i'}\in S_i(\Delta_{n,m})\}|}{n-1}
    -|S_i(\Delta_{n,m})|
    \right|
    \ge \epsilon
\right)  \,\le\,
2n(nm)^{-\frac54(1+\gamma)}  \,\le\,
2(nm)^{-(\tilde C_1+\gamma)},
\end{aligned}
\end{equation}
where $\tilde C_1=(1+\gamma)/4>0$.
Thus, with probability at least 
\(
1 - 2(nm)^{-(\tilde C_1 + \gamma)},
\)
we have
\begin{equation*}
\begin{aligned}
    \min_i \frac{|\{\;i' \neq i : \xi_{i'} \in S_i(\Delta_{n,m})\}|}{n-1} \;\;\ge\;\; \min_i |S_i(\Delta_{n,m})| \;-\;
    \epsilon &\;\ge\; \Delta_{n,m} \;-\; C_2 \,\sqrt{\frac{\log nm}{nmh_1}} \\[3pt]
    &\;=\; C_1 \,\sqrt{\frac{\log nm}{nmh_1}}
\end{aligned}
\end{equation*}
as claimed.
\end{proof}
\end{lemma}

\begin{lemma} \label{lem4}
Let $C_1, C_2>0$ be the global constants from Lemma~\ref{lem3}, and let $C_0$ be any constant satisfying $0 < C_0 \le C_1$. Define  $h_2 = C_0\, \sqrt{\tfrac{\log nm}{nmh_1}},$
and let $\N^{(t)}_i$ denote the neighborhood obtained by thresholding at the lower $h_2$-quantile of $\{\tilde d_t(i,k)\}$ as in~\eqref{eq:nbhd}. Choose a sufficiently large constant $C_3>0$, depending only
on $B_0$, $B_2$, and $\gamma$. Then, for all sufficiently large $n,m$ satisfying the conditions of Lemma~\ref{lem3} and $\rho_n > \sqrt{\tfrac{mh_1 \log nm}{n}}$, there exists $\Tilde C_2 > 0$ such that 
\begin{enumerate}
    \item 
    $|\mathcal{N}^{(t)}_i| 
        \;\ge\; 
        C_0\,\sqrt{\frac{n\log nm}{mh_1}}$
    for some positive constant $C_0$,
    \item 
    for any $i$ and  any $i' \in \N^{(t)}_i$,
    $\|P^{(t)}_{i'.} - P^{(t)}_{i.}\|_2^2 \,/\, n 
        \;\;\le\;\; 
        (6L_2(C_1+C_2) + 12C_3)\rho_n \sqrt{\frac{\log nm}{nmh_1}} + 24B_0B_1\rho_n^2 h_1^{\beta}$
        \, \,  with the probability
$1 - 2\,(nm)^{-(\widetilde C_1+\gamma)} - 2\,(nm)^{-(\widetilde C_2+\gamma)}.$
\end{enumerate}
\end{lemma}

\begin{proof}
To control the target inequality, we first bound the distance $\tilde d_t(i,i')$ between a node $i$ and one of its neighbors $i'$. This requires an upper bound on the difference $|(\Tilde P^{(t)^{2}}/n)_{ij} - (P^{(t)^2}/n)_{ij}|$.
\begin{align*}
    |(\Tilde P^{(t)^2}/n)_{ij} - (P^{(t)^2}/n)_{ij}| &=|\sum_l (\Tilde P^{(t)}_{il}\Tilde P^{(t)}_{lj} -  P^{(t)}_{il}P^{(t)}_{lj})|/n\\
    & \leq  \frac{|\sum_{l \neq i,j}(\Tilde P^{(t)}_{il}\Tilde P^{(t)}_{lj} -  P^{(t)}_{il}P^{(t)}_{lj})|}{n-2} \frac{n-2}{n}\\[8pt]
    &\quad \quad+ \frac{|(\Tilde P^{(t)}_{ii}\Tilde P^{(t)}_{ij} - P^{(t)}_{ii}P^{(t)}_{ij})|}{n} 
    + \frac{|(\Tilde P^{(t)}_{ij}\Tilde P^{(t)}_{jj} - P^{(t)}_{ij}P^{(t)}_{jj})|}{n}
\end{align*}
\noindent
The first term of the inequality can be divided into variance and bias terms.
\begin{align*}
   \frac{|\sum_{l \neq i,j}(\Tilde P^{(t)}_{il}\Tilde P^{(t)}_{lj} -  P^{(t)}_{il}P^{(t)}_{lj})|}{n-2} \leq
    \frac{|\sum_{l \neq i,j}(\Tilde P^{(t)}_{il}\Tilde P^{(t)}_{lj} -  {\rm E} \Tilde P^{(t)}_{il} \Tilde P^{(t)}_{lj})|}{n-2} + \frac{|\sum_{l \neq i,j}( {\rm E} \Tilde P^{(t)}_{il}\Tilde P_{lj}^{(t)} -  P^{(t)}_{il}P^{(t)}_{lj})|}{n-2}
\end{align*}
For the following moment calculations, we condition on the latent positions $\xi_1,\ldots,\xi_n$ and suppress this conditioning in the notation. The kernel estimators $\Tilde P_{il}^{(t)}$ and
$\Tilde P_{lj}^{(t)}$ are independent when $i \neq j$ and $l \neq i,j$. By the absolute weight bound in Lemma 1,
\[
    \left|{\rm E}\Tilde P_{ab}^{(t)}\right|
    =
    \left|\sum_{k=1}^m w_k(t)P_{ab}^{(t_k)}\right|
    \leq B_0\rho_n,
    \qquad a \neq b.
\]
Moreover, ${\rm Var}(\Tilde P_{ab}^{(t)})
\leq B_2\rho_n/(mh_1)$ by (\ref{lem2:var}).
Consequently, for any $l \neq i,j$,
\begin{equation} \label{eq:var1}
\begin{aligned}
{\rm Var}\bigl(\Tilde P_{il}^{(t)}\Tilde P_{lj}^{(t)}\bigr)
&=
{\rm Var}(\Tilde P_{il}^{(t)})
{\rm Var}(\Tilde P_{lj}^{(t)})\\
&\quad +  \bigl({\rm E}\Tilde P_{il}^{(t)}\bigr)^2
{\rm Var}(\Tilde P_{lj}^{(t)})
+\bigl({\rm E}\Tilde P_{lj}^{(t)}\bigr)^2
{\rm Var}(\Tilde P_{il}^{(t)})\\
&\leq
\frac{B_2^2\rho_n^2}{(mh_1)^2}+\frac{2B_0^2B_2\rho_n^3}{mh_1}\\
&\leq
\frac{C_V\rho_n^2}{mh_1}.
\end{aligned}
\end{equation}
The equality follows from independence. The last inequality uses $mh_1 \geq 1/2$, as required
for applying Lemma 1, and $\rho_n \leq 1$, with
$C_V = 2B_2^2 + 2B_0^2B_2$.

\noindent
By the absolute weight bound in Lemma 1,
$|\Tilde P_{ab}^{(t)}| \leq B_0$.
For $i \neq j$ and $l \neq i,j$, define
\[
    Z_{ijl}^{(t)}
    =
    \Tilde P_{il}^{(t)}\Tilde P_{lj}^{(t)}
    -
    {\rm E}\bigl(
        \Tilde P_{il}^{(t)}\Tilde P_{lj}^{(t)}
    \bigr).
\]
These variables are independent across $l$, have mean zero, and satisfy
\[
    |Z_{ijl}^{(t)}| \leq 2B_0^2,
    \qquad
    {\rm Var}(Z_{ijl}^{(t)})
    \leq \frac{C_V\rho_n^2}{mh_1}
\]
by (\ref{eq:var1}). Bernstein's inequality, followed by a union bound over all $i \neq j$, therefore gives
\begin{align*}
    &{\rm Pr}\left(
    \max_{i \neq j}
    \left|\frac{1}{n-2}\sum_{l \neq i,j} Z_{ijl}^{(t)}\right|\geq \epsilon\right)
    \,\leq\,
    2n^2 \exp\left\{-\frac{(n-2)\epsilon^2}{2\left(C_V\rho_n^2/(mh_1)+2B_0^2\epsilon/3\right)}    \right\}.
\end{align*}
For $n \geq 4$ and any fixed $a>0$, the right-hand side is at most $2(nm)^{-a}$ provided that $\epsilon$ satisfies the following bound, with a sufficiently large implicit constant depending only on $a$, $B_0$, and $C_V$:
\begin{equation}\label{eq:epsilon}
\begin{aligned}
    \epsilon
    \,\gtrsim\,
    \frac{\log nm}{n}
    \quad\vee\quad
    \rho_n\sqrt{\frac{\log nm}{nmh_1}}.
\end{aligned}
\end{equation}
Under the sparsity condition in Lemma~\ref{lem4},
$\rho_n > \sqrt{mh_1\log nm/n}$, we set
\begin{align*}
    \epsilon
    =
    C_3\rho_n\sqrt{\frac{\log nm}{nmh_1}}.
\end{align*}
This condition implies
$\epsilon \leq C_3\rho_n^2/(mh_1)$.
Hence, for $n \geq 4$, the exponent in the preceding
Bernstein bound satisfies
\begin{align*}
    \frac{(n-2)\epsilon^2}{
        2\left(
            C_V\rho_n^2/(mh_1)
            +2B_0^2\epsilon/3
        \right)
    }
    &\geq
    \frac{C_3^2}{
        4C_V+(8/3)B_0^2C_3
    }\log nm.
\end{align*}
Define $\kappa_3={C_3^2/\left(4C_V+(8/3)B_0^2C_3\right)}.$
Choose $C_3$ sufficiently large that $\kappa_3>\gamma+3$,
and set $\Tilde C_2=\kappa_3-\gamma-3>0$.
It follows that, for sufficiently large $n,m$,
\begin{align*}
    &{\rm Pr}\left(
        \max_{i \neq j} \left|
            \frac{1}{n-2}
            \sum_{l \neq i,j} Z_{ijl}^{(t)}
        \right| \geq C_3\rho_n\sqrt{\frac{\log nm}{nmh_1}}
    \right) \leq
    2n^2(nm)^{-\kappa_3} \leq
    \frac{2}{3}(nm)^{-(\Tilde C_2+\gamma)}.
\end{align*}
For the two diagonal-related terms, the zero-diagonal
convention gives
\begin{align*}
    \Tilde P_{ii}^{(t)}
    = \sum_{k=1}^m w_k(t)A_{ii}^{(t_k)}
    = 0,
    \qquad
    P_{ii}^{(t)} = 0.
\end{align*}
Therefore, for any $i \neq j$,
\begin{equation}
\begin{aligned}
    \Tilde P_{ii}^{(t)}\Tilde P_{ij}^{(t)}
    - P_{ii}^{(t)}P_{ij}^{(t)}
    = 0, \quad \Tilde P_{ij}^{(t)}\Tilde P_{jj}^{(t)}
    - P_{ij}^{(t)}P_{jj}^{(t)} = 0.
\end{aligned}
\end{equation}
Thus, the two diagonal contributions to $(\Tilde P^{(t)})^2-(P^{(t)})^2$ vanish exactly.


\noindent
Lastly, for any $i \neq j$ and $l \neq i,j$, independence,
(\ref{lem2:bias}), and the bounds
$|{\rm E}\Tilde P_{il}^{(t)}| \leq B_0\rho_n$
and $0 \leq P_{lj}^{(t)} \leq \rho_n$ give
\begin{align*}
    \left|{\rm E}\bigl(\Tilde P_{il}^{(t)}\Tilde P_{lj}^{(t)}\bigr)
    - P_{il}^{(t)}P_{lj}^{(t)}\right|
    &= \left|\bigl({\rm E}\Tilde P_{il}^{(t)}\bigr)b_{lj}(t)
    + P_{lj}^{(t)}b_{il}(t)\right| \\
    &\leq B_0\rho_n|b_{lj}(t)| + \rho_n|b_{il}(t)| \\
    &\leq (B_0+1)B_1\rho_n^2 h_1^\beta
    \leq 2B_0B_1\rho_n^2 h_1^\beta.
\end{align*}
The last inequality uses $B_0 \geq 1$, which follows from
\begin{align*}
    1
    = \left|\sum_{k=1}^m w_k(t)\right|
    \leq \sum_{k=1}^m |w_k(t)|
    \leq B_0
\end{align*}
by Lemma~\ref{lem:weights}.
Combining this bias bound with the preceding concentration
bound, the vanishing diagonal terms, and $(n-2)/n \leq 1$,
we obtain
\begin{align}
\label{estiP:square}
    \max_{i \neq j}
    \left|
        \bigl((\Tilde P^{(t)})^2/n\bigr)_{ij} - \bigl((P^{(t)})^2/n\bigr)_{ij}
    \right|
    &\leq C_3\rho_n\sqrt{\frac{\log nm}{nmh_1}}
    + 2B_0B_1\rho_n^2 h_1^\beta.
\end{align}
with probability at least $1-2(nm)^{-(\Tilde C_2+\gamma)}$.

Fix $i$ and $\Tilde i \neq i$ such that
$\xi_{\Tilde i} \in S_i(\Delta_{n,m})$.
By Assumption~\ref{assump:lipschitz},
$|P_{ik}^{(t)}-P_{\Tilde i k}^{(t)}|
\leq L_2\rho_n\Delta_{n,m}$ for $k \notin \{i,\Tilde i\}$.
Since $P^{(t)}$ is symmetric with zero diagonal, the two
terms corresponding to $k=i$ and $k=\Tilde i$ in the
matrix-product difference combine to
$P_{i\Tilde i}^{(t)}
(P_{\Tilde i l}^{(t)}-P_{il}^{(t)})$.
Moreover,
$|P_{\Tilde i l}^{(t)}-P_{il}^{(t)}|
\leq L_2\rho_n\Delta_{n,m}$ when $l \notin \{i,\Tilde i\}$,
and this difference is at most $\rho_n$ in absolute value
when $l \in \{i,\Tilde i\}$.
Therefore, for all $l=1,\ldots,n$,
\begin{equation}\label{trueP:squre}
\begin{aligned}
    \left|
        \bigl((P^{(t)})^2/n\bigr)_{il}
        - \bigl((P^{(t)})^2/n\bigr)_{\Tilde i l}
    \right|
    &\leq \frac{n-2}{n}L_2\rho_n^2\Delta_{n,m}
    + \frac{\rho_n}{n}
    \left|P_{\Tilde i l}^{(t)}-P_{il}^{(t)}\right| \\
    &\leq L_2\rho_n^2\Delta_{n,m}
    + \frac{\rho_n^2}{n}
    \mathbf{1}\{l \in \{i,\Tilde i\}\}.
\end{aligned}
\end{equation}
We bound the distance between node $i$ and its proxy $\Tilde i$, chosen so that $\xi_{\Tilde i}\in S_i(\Delta_{n,m})$. From the earlier results, we have, with probability at least $1-2(nm)^{-(\Tilde C_2+\gamma)}$,
\begin{equation}\label{iitilde}
\begin{aligned}
    \Tilde d_t^2(i,\Tilde i)
    &= \max_{l \notin \{i,\Tilde i\}}
    \left|
        \bigl((\Tilde P^{(t)})^2/n\bigr)_{il}
        - \bigl((\Tilde P^{(t)})^2/n\bigr)_{\Tilde i l}
    \right| \\
    &\leq \max_{l \notin \{i,\Tilde i\}}
    \left|
        \bigl((P^{(t)})^2/n\bigr)_{il}
        - \bigl((P^{(t)})^2/n\bigr)_{\Tilde i l}
    \right|
    + 2\max_{i \neq j}
    \left|
        \bigl((\Tilde P^{(t)})^2/n\bigr)_{ij}
        - \bigl((P^{(t)})^2/n\bigr)_{ij}
    \right| \\
    &\leq L_2\rho_n^2\Delta_{n,m}
    + 2C_3\rho_n\sqrt{\frac{\log nm}{nmh_1}}
    + 4B_0B_1\rho_n^2 h_1^\beta.
\end{aligned}
\end{equation}
for all $i$ and $\Tilde i$ such that $\xi_{\Tilde i} \in S_i(\Delta_{n,m})$. The transition to the last line is from (\ref{estiP:square}) and (\ref{trueP:squre}).

\noindent
Recalling $\Delta_{n,m} = (C_1+C_2)\sqrt{\tfrac{\log nm}{nmh_1}}$, by Lemma~\ref{lem3} and (\ref{estiP:square}), for all $i$, at least $C_1\sqrt{\tfrac{\log nm}{nmh_1}}$ fraction of nodes $\Tilde i \neq i$ satisfy both $\xi_{\Tilde i} \in S_i(\Delta_{n,m})$ and
\begin{equation}\label{eq:diitilde}
    \Tilde d_t^2(i,\Tilde i)
    \leq \bigl(L_2(C_1+C_2)+2C_3\bigr)\rho_n
    \sqrt{\frac{\log nm}{nmh_1}}
    + 4B_0B_1\rho_n^2 h_1^\beta.
\end{equation}
Since $\N_i^{(t)}$ is obtained by thresholding the distances
at their lower $h_2$-quantile, where
$h_2 = C_0\sqrt{\tfrac{\log nm}{nmh_1}}
\leq C_1\sqrt{\tfrac{\log nm}{nmh_1}}$,
(\ref{eq:diitilde}) yields
\begin{align}\label{ii'}
    \Tilde d_t^2(i,i')
    &\leq \bigl(L_2(C_1+C_2)+2C_3\bigr)\rho_n
    \sqrt{\frac{\log nm}{nmh_1}}
    + 4B_0B_1\rho_n^2 h_1^\beta
\end{align}
for all $i$ and all $i' \in \N_i^{(t)}$, with probability
at least $1-2(nm)^{-(\Tilde C_1+\gamma)} -2(nm)^{-(\Tilde C_2+\gamma)}$.
For any $i$ and $i' \in \N_i^{(t)}$ with $i' \neq i$,
on the event in Lemma~\ref{lem3} and for sufficiently large
$n,m$, choose proxy indices $\Tilde i$ and $\Tilde i'$
such that $i,i',\Tilde i,\Tilde i'$ are pairwise distinct,
$\xi_{\Tilde i} \in S_i(\Delta_{n,m})$, and
$\xi_{\Tilde i'} \in S_{i'}(\Delta_{n,m})$.
The case $i'=i$ is immediate.
Then
\begin{equation}\label{lem4:final}
\begin{aligned}
    \|P_{i'.}^{(t)}-P_{i.}^{(t)}\|_2^2/n
    &= \bigl((P^{(t)})^2/n\bigr)_{ii}
    - \bigl((P^{(t)})^2/n\bigr)_{i'i}
    + \bigl((P^{(t)})^2/n\bigr)_{i'i'}
    - \bigl((P^{(t)})^2/n\bigr)_{ii'} \\
    &\leq
    \left|\bigl((P^{(t)})^2/n\bigr)_{i\Tilde i}
    - \bigl((P^{(t)})^2/n\bigr)_{i'\Tilde i}\right|
    + \left|\bigl((P^{(t)})^2/n\bigr)_{i'\Tilde i'}
    - \bigl((P^{(t)})^2/n\bigr)_{i\Tilde i'}\right| \\
    &\quad + 4L_2\rho_n^2\Delta_{n,m}
    + \frac{2\rho_n^2}{n} \\
    &\leq 2\Tilde d_t^2(i,i')
    + 4L_2\rho_n^2\Delta_{n,m}
    + 8C_3\rho_n\sqrt{\frac{\log nm}{nmh_1}}
    + 16B_0B_1\rho_n^2h_1^\beta \\
    &\leq \bigl(6L_2(C_1+C_2)+12C_3\bigr)\rho_n
    \sqrt{\frac{\log nm}{nmh_1}}
    + 24B_0B_1\rho_n^2h_1^\beta
\end{aligned}
\end{equation}
with probability at least
$1-2(nm)^{-(\Tilde C_1+\gamma)}
  -2(nm)^{-(\Tilde C_2+\gamma)}$.

The first inequality uses (\ref{trueP:squre}) four times.
Two of these applications involve a diagonal entry, giving
the additional term $2\rho_n^2/n$.
Replacing the four remaining matrix-product entries by
their empirical counterparts using (\ref{estiP:square})
contributes at most
$4C_3\rho_n\sqrt{\tfrac{\log nm}{nmh_1}}
+8B_0B_1\rho_n^2h_1^\beta$.
The diagonal correction is absorbed into the displayed
concentration term for sufficiently large $n,m$, since
\begin{align*}
    \rho_n\sqrt{\frac{\log nm}{nmh_1}}
    > \frac{\log nm}{n},
    \qquad
    \frac{2\rho_n^2}{n}
    \leq 4C_3\rho_n\sqrt{\frac{\log nm}{nmh_1}}
\end{align*}
under the sparsity condition of Lemma~\ref{lem4}. Finally, (\ref{ii'}), the definition of $\Delta_{n,m}$, and $\rho_n \leq 1$ give the last inequality. This completes the proof.
\end{proof}


\noindent
Based on Lemma~\ref{lem3} and \ref{lem4}, we are ready to prove Theorem~\ref{thm:double-smooth}, which provides the error bound for the double-smoothing estimator.
 \begin{proof}
Since $\hat P_{ii}^{(t)}=P_{ii}^{(t)}=0$, we have
\begin{equation}
\begin{aligned}
    \frac{1}{n}\sum_j
    \bigl(\hat P_{ij}^{(t)}-P_{ij}^{(t)}\bigr)^2
    &= \frac{1}{n}\sum_{j\neq i}
    \left\{
        \frac{\sum_{i'\in\N_i^{(t)}}
        \bigl(\Tilde P_{i'j}^{(t)}-P_{ij}^{(t)}\bigr)}
        {|\N_i^{(t)}|}
    \right\}^2 \\
    &\leq \frac{2}{n}\sum_j
    \left\{
        \frac{\sum_{i'\in\N_i^{(t)}}
        \bigl(\Tilde P_{i'j}^{(t)}-P_{i'j}^{(t)}\bigr)}
        {|\N_i^{(t)}|}
    \right\}^2
    + \frac{2}{n}\sum_j
    \left\{
        \frac{\sum_{i'\in\N_i^{(t)}}
        \bigl(P_{i'j}^{(t)}-P_{ij}^{(t)}\bigr)}
        {|\N_i^{(t)}|}
    \right\}^2 \\
    &= \frac{2}{n}\sum_j J_1(i,j)
    + \frac{2}{n}\sum_j J_2(i,j).
\end{aligned}
\end{equation}
Here $J_1(i,j)$ and $J_2(i,j)$ denote the first and second
squared averages, respectively, for all $j$.
The inequality uses $(a+b)^2\leq 2a^2+2b^2$ and adds
the nonnegative terms corresponding to $j=i$.

Firstly, we need to bound $J_1$.
\begin{equation}
\begin{aligned}
\frac{1}{n}\sum_j J_1(i,j)
&= \frac{1}{n|\N^{(t)}_i|^2}\sum_{j}
\bigg\{ \sum_{i' \in \N^{(t)}_i} (\Tilde P_{i'j}^{(t)} - P_{i'j}^{(t)})\bigg\}^2 \\
&= \frac{1}{n|\N^{(t)}_i|^2}\sum_{j}
\bigg\{
\sum_{i' \in \N^{(t)}_i} (\Tilde P_{i'j}^{(t)} - P_{i'j}^{(t)})^2 \\
&\qquad\quad
+ \sum_{i' \in \N^{(t)}_i}
\sum_{\substack{i'' \in \N^{(t)}_i \\ i'' \neq i'}}
(\Tilde P_{i'j}^{(t)} - P_{i'j}^{(t)})
(\Tilde P_{i''j}^{(t)} - P_{i''j}^{(t)})
\bigg\} \\
&= \frac{1}{|\N^{(t)}_i|^2}
\sum_{i' \in \N^{(t)}_i}
\bigg\{
\frac{1}{n}\sum_{j} (\Tilde P_{i'j}^{(t)} - P_{i'j}^{(t)})^2 \\
&\qquad\quad
+ \frac{1}{n}\sum_{j}
\sum_{\substack{i'' \in \N^{(t)}_i \\ i'' \neq i'}}
(\Tilde P_{i'j}^{(t)} - P_{i'j}^{(t)})
(\Tilde P_{i''j}^{(t)} - P_{i''j}^{(t)})
\bigg\}
\end{aligned}
\end{equation}
For the first term, let
$X_{i'j}^{(t)}=\Tilde P_{i'j}^{(t)}-P_{i'j}^{(t)}$
and $Y_{i'j}^{(t)}=(X_{i'j}^{(t)})^2$.
Set
\begin{align*}
    M=B_0+1,
    \qquad
    q_n=\frac{B_2\rho_n}{mh_1}
    +B_1^2\rho_n^2h_1^{2\beta}.
\end{align*}
The absolute-weight bound gives
$0\leq Y_{i'j}^{(t)}\leq M^2$.
Moreover, (\ref{lem2:bias}) and (\ref{lem2:var}) imply
${\rm E}Y_{i'j}^{(t)}\leq q_n$ for $j\neq i'$,
and $Y_{i'i'}^{(t)}=0$.
Since $e^x\leq 1+(e-1)x$ for $0\leq x\leq 1$,
\begin{equation}
\begin{aligned}
    {\rm E}\exp\left(\frac{Y_{i'j}^{(t)}}{M^2}\right)
    &\leq 1+\frac{e-1}{M^2}{\rm E}Y_{i'j}^{(t)} \\
    &\leq \exp\left(\frac{(e-1)q_n}{M^2}\right).
\end{aligned}
\end{equation}
Independence within each row, Markov's inequality,
and a union bound over rows therefore give, for any $u>0$,
\begin{equation} \label{eq:markovbound}
    {\rm Pr}\left(
        \max_{i'}\frac{1}{n}\sum_j Y_{i'j}^{(t)}
        > (e-1)q_n+\frac{M^2u}{n}
    \right)
    \leq n e^{-u}.
\end{equation}
Fix $\Tilde C_3>0$ and take
$u=(\Tilde C_3+\gamma+1)\log nm$.
Then
\begin{equation} \label{eq:markovboundres}
\begin{aligned}
    \max_{i'}\frac{1}{n}\sum_j
    \bigl(\Tilde P_{i'j}^{(t)}-P_{i'j}^{(t)}\bigr)^2
    &\leq (e-1)\left(
        \frac{B_2\rho_n}{mh_1}
        +B_1^2\rho_n^2h_1^{2\beta}
    \right) \\
    &\quad +(B_0+1)^2(\Tilde C_3+\gamma+1)
    \frac{\log nm}{n}
\end{aligned}
\end{equation}
with probability at least
$1-(nm)^{-(\Tilde C_3+\gamma)}$,
since $n e^{-u}\leq(nm)^{-(\Tilde C_3+\gamma)}$.

\noindent
For the second term,
\begin{equation}
\begin{aligned}
    \frac{1}{n|\N^{(t)}_i|^2}&\sum_{j}\sum_{i' \in \N_i}\sum_{i'' \neq i', i'' \in \N^{(t)}_i }(\Tilde P_{i'j}^{(t)} - P_{i'j}^{(t)})(\Tilde P_{i''j}^{(t)} - P_{i''j}^{(t)})\\[5pt]
    & \leq \frac{1}{|\N^{(t)}_i|^2}\sum_{i' \in \N^{(t)}_i}\sum_{i'' \neq i', i'' \in \N^{(t)}_i }\Bigg|\frac{1}{n}\sum_{j}(\Tilde P_{i'j}^{(t)} - P_{i'j}^{(t)})(\Tilde P_{i''j}^{(t)} - P_{i''j}^{(t)}) \Bigg|\\[5pt]
    &\leq \frac{1}{|\N^{(t)}_i|^2}\sum_{i' \in \N^{(t)}_i}\sum_{i'' \neq i', i'' \in \N^{(t)}_i } \Bigg\{ \Bigg| \frac{1}{n-2} \sum_{j \neq i', i''}    (\Tilde P_{i'j}^{(t)} - P_{i'j}^{(t)})(\Tilde P_{i''j}^{(t)} - P_{i''j}^{(t)})\Bigg|\frac{n-2}{n} \ + \\[5pt]
    \quad &\frac{|(\Tilde P_{i'i'}^{(t)} - P_{i'i'}^{(t)})(\Tilde P_{i''i'}^{(t)} - P_{i''i'}^{(t)})|}{n} + \frac{|(\Tilde P_{i'i''}^{(t)} - P_{i'i''}^{(t)})(\Tilde P_{i''i''}^{(t)} - P_{i''i''}^{(t)})|}{n} \Bigg\}
\end{aligned}
\end{equation}
The two exceptional terms in \textup{(23)} vanish because
$X_{i'i'}^{(t)}=X_{i''i''}^{(t)}=0$.
For distinct $i'$ and $i''$, and for $j\neq i',i''$, define
\[
W_j^{(t)}
=
X_{i'j}^{(t)}X_{i''j}^{(t)}
-
b_{i'j}(t)b_{i''j}(t),
\]
where $X_{i'j}^{(t)}$, $M$, and $q_n$ are as defined above,
and $b_{i'j}(t)={\rm E}X_{i'j}^{(t)}$.
By independence,
\begin{equation}
\begin{aligned}
{\rm E}\bigl(X_{i'j}^{(t)}X_{i''j}^{(t)}\bigr)
&=
b_{i'j}(t)b_{i''j}(t),\\
{\rm Var}\bigl(W_j^{(t)}\bigr)
&\le
{\rm E}\bigl[(X_{i'j}^{(t)})^2\bigr]\,
{\rm E}\bigl[(X_{i''j}^{(t)})^2\bigr]
\le q_n^2.
\end{aligned}
\end{equation}
The variables $W_j^{(t)}$ are independent across
$j\neq i',i''$, have mean zero, and satisfy
$|W_j^{(t)}|\le 2M^2$.
Thus, Bernstein's inequality gives, for any $\epsilon>0$,
\begin{equation}
\begin{aligned}
{\rm Pr}\left( \left| \frac{1}{n-2}\sum_{j\neq i',i''}W_j^{(t)}
\right|>\epsilon \right)
&\le 2\exp\left( -\frac{(n-2)\epsilon^2} {2q_n^2+\frac{4}{3}M^2\epsilon} \right).
\end{aligned}
\end{equation}
Set $u=(\Tilde C_3+\gamma+3)\log nm$ and
\[
\epsilon = q_n\sqrt{\frac{2u}{n-2}} + \frac{4M^2u}{3(n-2)}.
\]
Then the preceding probability is at most $2e^{-u}$.
Taking a union bound over all $i'\neq i''$ gives
a failure probability at most $2n(n-1)e^{-u}\le (nm)^{-(\Tilde C_3+\gamma)}$ for sufficiently large $n,m$. Moreover, \eqref{lem2:bias} gives $|b_{i'j}(t)b_{i''j}(t)| \le B_1^2\rho_n^2h_1^{2\beta}$.
Consequently,
\begin{equation}
\begin{aligned}
\max_{i'\neq i''}
\left|
\frac{1}{n-2}
\sum_{j\neq i',i''}
X_{i'j}^{(t)}X_{i''j}^{(t)}
\right|
\le
\epsilon+B_1^2\rho_n^2h_1^{2\beta},
\end{aligned}
\end{equation}
with probability at least $1-(nm)^{-(\Tilde C_3+\gamma)}$.

\noindent
Let $k_i=|\N_i^{(t)}|$.
For fixed $C_0>0$, the lower-quantile construction gives
\begin{align*}
    k_i \geq (n-1)h_2
    \geq \frac{C_0}{2}
    \sqrt{\frac{n\log nm}{mh_1}}
\end{align*}
for $n\geq 2$.
Combining the preceding row-error and cross-product bounds, with probability at least $1-2(nm)^{-(\Tilde C_3+\gamma)}$, we obtain simultaneously for all $i$,
\begin{equation}
\begin{aligned}
    \frac{1}{n}\sum_j J_1(i,j)
    &\leq \frac{1}{k_i}
    \left\{
        (e-1)q_n
        +M^2(\Tilde C_3+\gamma+1)\frac{\log nm}{n}
    \right\} \\
    &\quad +\frac{k_i-1}{k_i}
    \left\{
        \epsilon+B_1^2\rho_n^2h_1^{2\beta}
    \right\} \\
    &\leq C_J\left\{
        \rho_n\sqrt{\frac{\log nm}{nmh_1}}
        +\rho_n^2h_1^\beta
    \right\}.
\end{aligned}
\end{equation}
Here $C_J>0$ depends only on $C_0,B_0,B_1,B_2,\gamma$, and $\Tilde C_3$. The first inequality uses the $k_i$ squared-error terms and the $k_i(k_i-1)$ cross-product terms in the expansion of $J_1$. For the last inequality, use the lower bound on $k_i$,
$mh_1\geq 1/2$, $0<h_1\leq 1$, and $\log nm\geq 1$. The sparsity condition also gives
\begin{align*}
    \frac{\log nm}{n}
    &< \rho_n\sqrt{\frac{\log nm}{nmh_1}}, \\
    \sqrt{\frac{\log nm}{n}}
    &< \frac{\rho_n}{\sqrt{mh_1}}
    \leq \sqrt{2}.
\end{align*}
These inequalities control the additional concentration terms and give the stated bound.

\noindent
Using Lemma~\ref{lem3} and \ref{lem4}, $J_2$ can be bounded as follows:
\begin{equation}
\begin{aligned}
    \frac{1}{n}\sum_j J_2(i,j)
    &= \frac{1}{n}\sum_j
    \left\{
        \frac{\sum_{i'\in\N_i^{(t)}}
        \bigl(P_{i'j}^{(t)}-P_{ij}^{(t)}\bigr)}
        {|\N_i^{(t)}|}
    \right\}^2 \\
    &\leq \frac{1}{|\N_i^{(t)}|}
    \sum_{i'\in\N_i^{(t)}}
    \frac{\|P_{i'.}^{(t)}-P_{i.}^{(t)}\|_2^2}{n} \\
    &\leq \bigl(6L_2(C_1+C_2)+12C_3\bigr)\rho_n
    \sqrt{\frac{\log nm}{nmh_1}}
    +24B_0B_1\rho_n^2h_1^\beta.
\end{aligned}
\end{equation}
with probability $1-2(nm)^{-(\Tilde C_1 +\gamma)}-2(nm)^{-(\Tilde C_2 +\gamma)}$.

\noindent
Then, combining the preceding bounds, with probability at least
$1-2(nm)^{-(\Tilde C_1+\gamma)} -2(nm)^{-(\Tilde C_2+\gamma)} -2(nm)^{-(\Tilde C_3+\gamma)}$, we have
\begin{equation}
\begin{aligned}
    \frac{1}{n}
    \|\hat P^{(t)}-P^{(t)}\|_{2,\infty}^2
    &\leq \frac{2}{n}\max_i
    \sum_j\bigl\{J_1(i,j)+J_2(i,j)\bigr\} \\
    &\leq C\rho_n\sqrt{\frac{\log nm}{nmh_1}}
    +B\rho_n^2h_1^\beta,
\end{aligned}
\end{equation}
where
\begin{align*}
    C &= 2\bigl(C_J+6L_2(C_1+C_2)+12C_3\bigr), \\
    B &= 2\bigl(C_J+24B_0B_1\bigr).
\end{align*}

To choose the optimal bandwidth $h_1$, we balance these two terms. Ignoring multiplicative constants for the purpose of finding the minimizer, we solve
\begin{equation}
\rho_n \sqrt{\frac{\log nm}{nmh_1}}=\rho_n^2 h_1^{\beta}.
\end{equation}
Equivalently,
\begin{equation}
\rho_n \sqrt{\frac{\log nm}{nm}}\; h_1^{-1/2}
= \rho_n^2 h_1^{\beta}.
\end{equation}
Rearranging terms yields
\begin{equation}
h_1^{2\beta + 1}
=
{\frac{\log nm}{\rho_n^2 nm}}.
\end{equation}
In particular, up to a multiplicative constant depending only on $C$ and $B$, we have
\begin{equation}
h_1^\star
\;\asymp\;
\Biggl(
\frac{\log nm}{\rho_n^2 nm}
\Biggr)^{\frac{1}{2\beta+1}}.
\end{equation}
Assume that Assumption~\ref{assump:LP} holds at the selected bandwidth
$h_1^\star$, with $1/(2m)\leq h_1^\star\leq 1$. Plugging this choice into the preceding error bound gives,
on the same event,
\begin{equation}
    \frac{1}{n}
    \|\hat P^{(t)}-P^{(t)}\|_{2,\infty}^2
    \leq C_\star
    \rho_n^{\frac{2\beta+2}{2\beta+1}}
    \left(\frac{\log nm}{nm}\right)^{\frac{\beta}{2\beta+1}},
\end{equation}
where $C_\star>0$ is a constant.
At this bandwidth, the sparsity condition
$\rho_n\gtrsim\sqrt{mh_1^\star\log nm/n}$
is equivalent, up to multiplicative constants, to
$\rho_n^2\gtrsim m^{\frac{\beta}{\beta+1}}\log nm/n$.

However, the condition $\rho_n^2\gtrsim m^{\frac{\beta}{\beta+1}}\log nm/n$ may fail to hold when $m$ grows much faster than $n$. In this regime, we explicitly set $\N_i^{(t)}=\{i\}$ for every $i$ and $t$. Consequently, the estimator reduces to the local polynomial estimator: $\hat P_{ij}^{(t)}=\Tilde P_{ij}^{(t)}$ for all node pairs $(i,j)$.

\noindent
With $\N_i^{(t)}=\{i\}$, we have $J_2(i,j)=0$, and
\begin{equation}
\begin{aligned}
    \frac{1}{n}\sum_j J_1(i,j)
    &= \frac{1}{n|\N_i^{(t)}|^2}\sum_j
    \left\{
        \sum_{i'\in\N_i^{(t)}}
        \bigl(\Tilde P_{i'j}^{(t)}-P_{i'j}^{(t)}\bigr)
    \right\}^2 \\
    &= \frac{1}{n}\sum_j
    \bigl(\Tilde P_{ij}^{(t)}-P_{ij}^{(t)}\bigr)^2.
\end{aligned}
\end{equation}
Define the stochastic component by $Z_{ij}^{(t)} = \Tilde{P}_{ij}^{(t)} - {\rm E}\Tilde{P}_{ij}^{(t)}$ and the bias component by $b_{ij}(t) = {\rm E}\Tilde{P}_{ij}^{(t)} - P_{ij}^{(t)}$.

\begin{equation}
|\Tilde{P}_{ij}^{(t)} - P_{ij}^{(t)}| \le |Z_{ij}^{(t)}| +  |b_{ij}(t)|.    
\end{equation}
The bias is bounded by $B_1 \rho_n h_1^\beta$ , \eqref{lem2:bias}. We need to control the stochastic term $Z_{ij}^{(t)}$ using Bernstein's inequality. The variance is ${\rm Var}(Z_{ij}^{(t)}) \le \frac{B_2 \rho_n}{mh_1}$.
Recall
\begin{equation}
Z_{ij}^{(t)} = \sum_k w_k(t) (A_{ij}^{(t_k)} - P_{ij}^{(t_k)}) 
\end{equation}
and the summands $w_k(t) (A_{ij}^{(t_k)} - P_{ij}^{(t_k)})$ are bounded by $M = \frac{B_0}{mh_1}$ (Lemma~\ref{lem:weights}). By Bernstein, we have
\begin{equation}
\begin{aligned}
    {\rm Pr}(|Z_{ij}^{(t)}| > \epsilon) &\,\le\, 2 \exp\left( -\frac{\epsilon^2/2}{\frac{B_2 \rho_n}{mh_1} + \frac{B_0}{mh_1} \frac{\epsilon}{3}} \right)\\[8pt]
    &\,\le\, 2 \max\!\left\{
    \exp\!\left(
        -\,\frac{\epsilon^{2}}{4}\,\frac{m h_{1}}{B_{2}\rho_{n}}
    \right),
    \;
    \exp\!\left(
        -\,\frac{3\epsilon}{4}\,\frac{m h_{1}}{B_{0}}
    \right)
\right\}
\end{aligned}
\end{equation}
Set
\begin{equation}
    \epsilon
    = B_4\left\{
        \sqrt{\frac{\rho_n\log nm}{mh_1}}
        \,\vee\,
        \frac{\log nm}{mh_1}
    \right\},
\end{equation}
where $B_4>0$ is chosen sufficiently large.
Writing
$\kappa_4=\min\{B_4^2/(4B_2),\,3B_4/(4B_0)\}$,
the preceding Bernstein bound gives
\begin{equation}
    {\rm Pr}\bigl(|Z_{ij}^{(t)}|>\epsilon\bigr)
    \le 2\exp(-\kappa_4\log nm)
    = 2(nm)^{-\kappa_4}.
\end{equation}
Choose $B_4$ such that $\kappa_4>\gamma+2$ and set
$\tilde B_0=\kappa_4-\gamma-2>0$. Then
\begin{equation}
    {\rm Pr}\bigl(|Z_{ij}^{(t)}|>\epsilon\bigr)
    \le 2(nm)^{-(\tilde B_0+\gamma+2)}.
\end{equation}

\noindent
Hence, by taking a union bound over $i\neq j$, we obtain
\begin{equation}
\begin{aligned}
    {\rm Pr}\left(
        \max_{i\neq j}|Z_{ij}^{(t)}|>\epsilon
    \right)
    &\le 2n^2(nm)^{-(\tilde B_0+\gamma+2)} \\
    &\le 2(nm)^{-(\tilde B_0+\gamma)}.
\end{aligned}
\end{equation}

\noindent
Using $(a+b)^{2} \le 2a^{2}+2b^{2}$, we obtain
\begin{equation}
\begin{aligned}
(\tilde P_{ij}^{(t)} - P_{ij}^{(t)})^{2}
&= (Z_{ij}^{(t)} + b_{ij}(t))^{2} \\
&\le 2\bigl(Z_{ij}^{(t)}\bigr)^{2} + 2\bigl(b_{ij}(t)\bigr)^{2} \\
&\le
2\epsilon^{2}
+
2 B_{1}^{2}\rho_{n}^{2}h_{1}^{2\beta}.
\end{aligned}
\end{equation}
\begin{equation}
    \bigl(\Tilde P_{ij}^{(t)}-P_{ij}^{(t)}\bigr)^2
    \le
    2B_4^2\left\{
        \frac{\rho_n\log nm}{mh_1}
        \,\vee\,
        \frac{(\log nm)^2}{(mh_1)^2}
    \right\}
    +2B_1^2\rho_n^2h_1^{2\beta},
\end{equation}
with probability at least $1 - 2(nm)^{-(\tilde B_{0}+\gamma)}$.

\noindent
With this error bound, choose the temporal bandwidth as
\begin{equation}
    h_1^\star
    \asymp
    \left(\frac{\log nm}{\rho_n m}\right)^{\frac{1}{2\beta+1}}.
\end{equation}
We apply Assumption 3 at this bandwidth, which satisfies $h_1^\star\geq 1/(2m)$ for sufficiently large $n,m$. Since $U(0)=(1,0,\ldots,0)^\top$, Assumption 3 gives
\begin{align*}
    \lambda_0
    &\leq U^\top(0)B_{mt}U(0)
    = \frac{1}{mh_1^\star}
      \sum_{k=1}^m
      K\left(\frac{t_k-t}{h_1^\star}\right)
    \leq \frac{K_{\max}}{h_1^\star}.
\end{align*}
Hence $h_1^\star\leq K_{\max}/\lambda_0$, and
\begin{align*}
    \frac{(\log nm)^2/(mh_1^\star)^2}
         {\rho_n\log nm/(mh_1^\star)}
    = \frac{\log nm}{\rho_n mh_1^\star}
    \asymp (h_1^\star)^{2\beta}
    \lesssim 1.
\end{align*}
Thus the additional Bernstein term $(\log nm/(mh_1^\star))^2$ is absorbed into $\rho_n\log nm/(mh_1^\star)$ up to a constant.

\noindent
In summary, set
$\tilde C=\frac12\min\{
\Tilde C_1,\Tilde C_2,\Tilde C_3,\tilde B_0
\}>0$.
For sufficiently large $n,m$, the preceding failure
probability bounds in either regime are at most
$2(nm)^{-(\tilde C+\gamma)}$.
Taking a sufficiently large constant $C>0$ independent
of $n$ and $m$, we obtain, for each fixed $t\in[0,1]$,
\begin{equation}
\begin{aligned}
    \frac{1}{n}
    \|\hat P^{(t)}-P^{(t)}\|_{2,\infty}^2
    \leq
    \begin{cases}
        C\rho_n^{\frac{2\beta+2}{2\beta+1}}
        \left(\dfrac{\log nm}{nm}\right)^{\frac{\beta}{2\beta+1}},
        & \text{if }
        n\gtrsim
        \rho_n^{-2}m^{\frac{\beta}{\beta+1}}\log nm,
        \\[6pt]
        C\rho_n^{\frac{2\beta+2}{2\beta+1}}
        \left(\dfrac{\log nm}{m}\right)^{\frac{2\beta}{2\beta+1}},
        & \text{otherwise},
    \end{cases}
\end{aligned}
\end{equation}
with probability at least $1-2(nm)^{-(\tilde C+\gamma)}$. This completes the proof of Theorem~\ref{thm:double-smooth}.
\end{proof}

\subsection{Proof of Theorem~\ref{thm:three-step}}

\begin{proof}
Throughout, set $h_3=h_1^\star$ and use the same kernel and local polynomial order as in the first stage. Write $w_k(t)=w_k(t;h_3)$ for the third-stage temporal weights. Since the bandwidth, kernel, and polynomial order are unchanged, Assumption~\ref{assump:LP} and Lemma~\ref{lem:weights} apply uniformly in $t\in[0,1]$. Using $\sum_{k=1}^m w_k(t)=1$, we start from the decomposition
\begin{equation}
\begin{aligned}
\bar P^{(t)}-P^{(t)}
&=
\sum_{k=1}^m w_k(t)
\bigl(\hat P^{(t_k)}-P^{(t_k)}\bigr)
+
\sum_{k=1}^m w_k(t)
\bigl(P^{(t_k)}-P^{(t)}\bigr).
\end{aligned}
\end{equation}
By the triangle inequality and $(a+b)^2\le 2a^2+2b^2$,
\begin{equation}
\begin{aligned}
\frac{1}{n}\|\bar P^{(t)}-P^{(t)}\|_{2,\infty}^2
&\le
\frac{2}{n}
\left\|\sum_{k=1}^m w_k(t)\bigl(\hat P^{(t_k)}-P^{(t_k)}\bigr)\right\|_{2,\infty}^2 + \frac{2}{n}
\left\| \sum_{k=1}^m w_k(t) \bigl(P^{(t_k)}-P^{(t)}\bigr)
\right\|_{2,\infty}^2.
\end{aligned}
\end{equation}

For the first term, the triangle inequality and the absolute-weight bound in Lemma~\ref{lem:weights} give
\begin{equation}
\begin{aligned}
&\sup_{t\in[0,1]}
\frac{2}{n}
\left\|
\sum_{k=1}^m w_k(t)
\bigl(\hat P^{(t_k)}-P^{(t_k)}\bigr)
\right\|_{2,\infty}^2 \\
&\qquad\le
\frac{2}{n}\sup_{t\in[0,1]}
\left(
\sum_{k=1}^m |w_k(t)|
\|\hat P^{(t_k)}-P^{(t_k)}\|_{2,\infty}
\right)^2
\le
\frac{2B_0^2}{n}
\max_{1\le k\le m}
\|\hat P^{(t_k)}-P^{(t_k)}\|_{2,\infty}^2.
\end{aligned}
\end{equation}

Applying a union bound over $t_1,\ldots,t_m$ to the concentration inequalities in the proof of Theorem~\ref{thm:double-smooth}, we obtain, in the first regime, 
\begin{equation}
\begin{aligned}
&{\rm Pr}\left(
\max_{1\le k\le m}
\frac{1}{n}
\|\hat P^{(t_k)}-P^{(t_k)}\|_{2,\infty}^2
>
C_\star
\rho_n^{\frac{2\beta+2}{2\beta+1}}
\left(\frac{\log nm}{nm}\right)^{\frac{\beta}{2\beta+1}}
\right) \,\le\, 2\sum_{r=1}^{3}(nm)^{-(\tilde C_r+\gamma)},
\end{aligned}
\end{equation}
for sufficiently large $n,m$ and a sufficiently large constant $C_\star>0$ independent of $n,m$.

Similarly, in the second regime, the same union-bound argument gives
\begin{equation}
\begin{aligned}
&{\rm Pr}\left(
\max_{1\le k\le m}
\frac{1}{n}
\|\hat P^{(t_k)}-P^{(t_k)}\|_{2,\infty}^2
> C_\star \rho_n^{\frac{2\beta+2}{2\beta+1}}
\left(\frac{\log nm}{m}\right)^{\frac{2\beta}{2\beta+1}}
\right) \,\le\, 2(nm)^{-(\tilde B_0+\gamma)},
\end{aligned}
\end{equation}
after increasing $C_\star$ if necessary.
Thus, taking
$\tilde C=\frac12\min\{\tilde C_1,\tilde C_2,\tilde C_3,\tilde B_0\}>0$,
there exists a constant $C_\star>0$, independent of $n,m$, such that, for sufficiently large $n,m$,
\begin{equation}
\begin{aligned}
\max_{1\le k\le m}
\frac{1}{n}
\|\hat P^{(t_k)}-P^{(t_k)}\|_{2,\infty}^2
\le
\begin{cases}
C_\star\rho_n^{\frac{2\beta+2}{2\beta+1}}
\bigl(\tfrac{\log nm}{nm}\bigr)^{\frac{\beta}{2\beta+1}},
& \text{in the first regime}, \\[6pt]
C_\star\rho_n^{\frac{2\beta+2}{2\beta+1}}
\bigl(\tfrac{\log nm}{m}\bigr)^{\frac{2\beta}{2\beta+1}},
& \text{in the second regime},
\end{cases}
\end{aligned}
\end{equation}
with probability at least $1-2(nm)^{-(\tilde C+\gamma)}$. Denote this simultaneous event by $\mathcal E$.

For the second term, since $h_3=h_1^\star$ and the weights coincide with those used in the first stage,
\begin{equation}
\begin{aligned}
\sum_{k=1}^m w_k(t)
\bigl(P_{ij}^{(t_k)}-P_{ij}^{(t)}\bigr)
&=
{\rm E}\Tilde P_{ij}^{(t)}-P_{ij}^{(t)}
=
b_{ij}(t).
\end{aligned}
\end{equation}
The bias bound \eqref{lem2:bias} therefore gives
\begin{equation}
\begin{aligned}
\sup_{t\in[0,1]}\max_{i,j}
\left|
\sum_{k=1}^m w_k(t)
\bigl(P_{ij}^{(t_k)}-P_{ij}^{(t)}\bigr)
\right|
\le B_1\rho_n h_3^\beta.
\end{aligned}
\end{equation}
The diagonal entries are zero, so this bound also holds when $i=j$. Consequently,
\begin{equation}
\begin{aligned}
\sup_{t\in[0,1]}
\frac{2}{n}
\left\|
\sum_{k=1}^m w_k(t)
\bigl(P^{(t_k)}-P^{(t)}\bigr)
\right\|_{2,\infty}^2
\le 2B_1^2\rho_n^2h_3^{2\beta}.
\end{aligned}
\end{equation}
Combining the two terms yields
\begin{equation}
\begin{aligned}
\sup_{t\in[0,1]}
\frac{1}{n}
\|\bar P^{(t)}-P^{(t)}\|_{2,\infty}^2
&\le
2B_0^2
\max_{1\le k\le m}
\frac{1}{n}
\|\hat P^{(t_k)}-P^{(t_k)}\|_{2,\infty}^2
+
2B_1^2\rho_n^2h_3^{2\beta}.
\end{aligned}
\end{equation}

It remains to bound the third-stage bias at $h_3=h_1^\star$. By the bandwidth choices in Theorem~\ref{thm:double-smooth}, there exists a fixed constant $c_h\ge 1$, independent of $n,m$, such that
\begin{equation}
\begin{aligned}
h_3=h_1^\star
\le
\begin{cases}
c_h\bigl(\tfrac{\log nm}{\rho_n^2nm}\bigr)^{\frac{1}{2\beta+1}},
& \text{in the first regime}, \\[6pt]
c_h\bigl(\tfrac{\log nm}{\rho_n m}\bigr)^{\frac{1}{2\beta+1}},
& \text{in the second regime}.
\end{cases}
\end{aligned}
\end{equation}
In the first regime, the condition
$\rho_n\ge\sqrt{\log nm/(nm)}$
implies $\log nm/(\rho_n^2nm)\le 1$. Hence,
\begin{equation}
\begin{aligned}
\rho_n^2h_3^{2\beta}
&\le
c_h^{2\beta}\rho_n^2
\left(\frac{\log nm}{\rho_n^2nm}\right)^{\frac{2\beta}{2\beta+1}} 
\le
c_h^{2\beta}\rho_n^2
\left(\frac{\log nm}{\rho_n^2nm}\right)^{\frac{\beta}{2\beta+1}} 
=
c_h^{2\beta}
\rho_n^{\frac{2\beta+2}{2\beta+1}}
\left(\frac{\log nm}{nm}\right)^{\frac{\beta}{2\beta+1}}.
\end{aligned}
\end{equation}
In the second regime,
\begin{equation}
\begin{aligned}
\rho_n^2h_3^{2\beta}
&\le
c_h^{2\beta}\rho_n^2
\left(\frac{\log nm}{\rho_n m}\right)^{\frac{2\beta}{2\beta+1}} = c_h^{2\beta}\rho_n^{\frac{2\beta+2}{2\beta+1}}
\left(\frac{\log nm}{m}\right)^{\frac{2\beta}{2\beta+1}}.
\end{aligned}
\end{equation}
Therefore, on $\mathcal E$, the asserted uniform error bound holds in both regimes with
$C=2B_0^2C_\star+2B_1^2c_h^{2\beta}$.
Since ${\rm Pr}(\mathcal E)\ge 1-2(nm)^{-(\tilde C+\gamma)}$, this completes the proof.
\end{proof}
\section{Smoothness sanity check} \label{app:smoothness-check}
We examine temporal variation in the U.S. Senate cosponsorship networks through a left--right smoothing diagnostic. At each boundary between years $t$ and $t+1$, we fit the two-stage estimator separately using all observations up to year $t$ and all observations from year $t+1$ onward. We require at least five years on each side, yielding 43 boundaries from 1977/78 to 2019/20. The polynomial degree and smoothing bandwidths are selected by cross-validation separately for the two fits.

Let $\hat P_L^{(t)}$ and $\hat P_R^{(t+1)}$ denote the fitted probability matrices at the last year on the left and the first year on the right, respectively. For this diagnostic, the endpoint estimates are symmetrized, truncated to $[0,1]$, and assigned zero diagonal entries. We define their contrast as
\begin{equation*}
\Delta^{(t)} = \hat P_R^{(t+1)} - \hat P_L^{(t)},
\end{equation*}
and evaluate $\|\Delta^{(t)}\|_F$ and $\|\Delta^{(t)}\|_{2,\infty}$. The Frobenius norm summarizes the overall contrast, while the $L_{2,\infty}$ norm captures the largest contrast in an individual state's connectivity profile.

Figure~\ref{fig:sanity-check} displays both contrasts across the eligible boundaries. The mean and mean plus two sample standard deviations, calculated separately for each norm across the 43 boundaries, are included as descriptive references. Although the left--right contrasts exhibit several peaks, these do not by themselves distinguish rapid continuous changes from abrupt change points.

\begin{figure}[htbp]
    \centering
    \includegraphics[width=\textwidth]{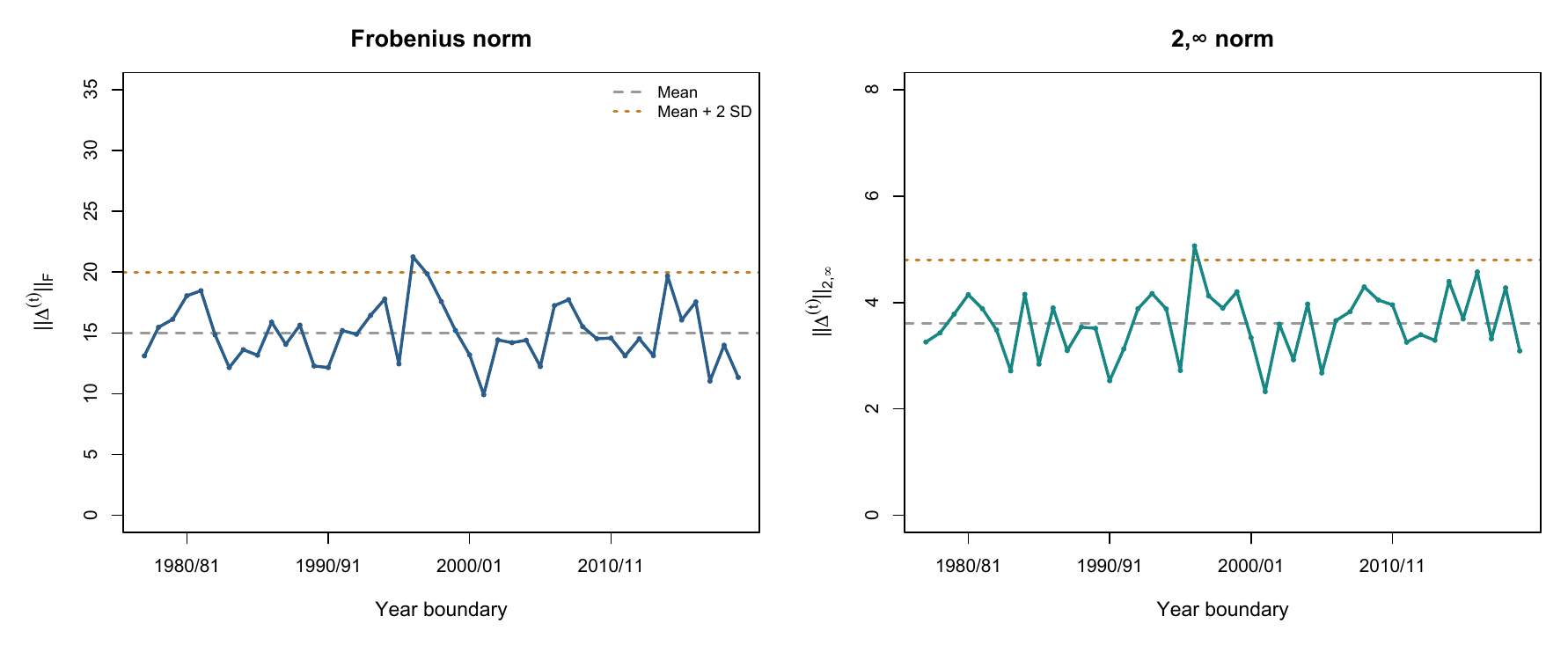}
    \caption{Left--right smoothing contrasts measured by the Frobenius norm (left) and the $L_{2,\infty}$ norm (right). The gray dashed lines indicate the mean across the 43 annual boundaries, and the orange dotted lines indicate the mean plus two sample standard deviations.}
    \label{fig:sanity-check}
\end{figure}
\end{appendix}

\end{document}